\documentclass{focs}

\def\final{1}

\ifdefined\final
\usepackage{microtype}
\newcommand{\todo}[1]{}
\newcommand{\onote}[1]{}
\newcommand{\hnote}[1]{}
\newcommand{\remove}[1]{}

\else
\usepackage{xcolor}
\newcommand{\todo}[1]{\begingroup\color{blue}TODO: #1\endgroup}
\newcommand{\onote}[1]{\begingroup\color{purple}Neil: #1\endgroup}
\newcommand{\hnote}[1]{\begingroup\color{purple}Nick: #1\endgroup}
\newcommand{\remove}[1]{\begingroup\color{gray}Removed: #1\endgroup}

\fi

\newcommand{\tcskip}{}
\newcommand{\iftc}[1]{}
\newcommand{\tcversion}[2]{#2}

\RequirePackage{latexsym}

\usepackage{array}
\usepackage{amsmath,amssymb,amsfonts}        
\usepackage{fancyhdr}
\usepackage{ifthen}                                 
\usepackage{float}
\usepackage{hyperref}
\usepackage{enumerate}
\usepackage{xcolor}

\allowdisplaybreaks[3]

\newcommand{\proofbelow}{5pt}

\definecolor{darkblue}{rgb}{0,0,0.38}
\definecolor{darkred}{rgb}{0.6,0,0}
\definecolor{darkgreen}{rgb}{0.1,0.35,0}
\hypersetup{colorlinks, linkcolor=darkblue, citecolor=darkgreen, urlcolor=darkblue}

\newcommand{\newterm}[1]{\textit{#1}}

\newcommand{\AppendixName}[1]{\label{app:#1}}
\newcommand{\ClaimName}[1]{\label{clm:#1}}
\newcommand{\ConjectureName}[1]{\label{conj:#1}}
\newcommand{\CorollaryName}[1]{\label{cor:#1}}

\newcommand{\EquationName}[1]{\label{eq:#1}}

\newcommand{\LemmaName}[1]{\label{lem:#1}}

\newcommand{\TheoremName}[1]{\label{thm:#1}}
\newcommand{\SectionName}[1]{\label{sec:#1}}

\newcommand{\Appendix}[1]{Appendix~\ref{app:#1}}
\newcommand{\Claim}[1]{Claim~\ref{clm:#1}}
\newcommand{\Conjecture}[1]{Conjecture~\ref{conj:#1}}
\newcommand{\Corollary}[1]{Corollary~\ref{cor:#1}}

\newcommand{\Lemma}[1]{Lemma~\ref{lem:#1}}

\newcommand{\Section}[1]{Section~\ref{sec:#1}}

\newcommand{\Theorem}[1]{Theorem~\ref{thm:#1}}

\newcommand{\afterproof}{$\qed$ \par \vspace{\proofbelow}}
\tcversion{\newenvironment{proofof}[1]{\vspace{\proofbelow} \par \textsc{Proof} \,(of #1).\,}{\afterproof}}{
    \newenvironment{proofof}[1]{\begin{proof}[Proof of~#1]}{\end{proof}}}

\newcommand{\union}{\cup}                                       
\newcommand{\intersect}{\cap}                                   
\newcommand{\abs}[1]{\lvert #1 \rvert}
\newcommand{\set}[1]{\left \{ #1 \right \}}                     
\newcommand{\setst}[2]{\left\{\, #1 \,:\, #2 \,\right\}}        
\newcommand{\card}[1]{\abs{#1}}
\newcommand{\smallsum}[2]{{\textstyle \sum_{#1}^{#2}}}
\newcommand{\smallprod}[2]{{\textstyle \prod_{#1}^{#2}}}
\newcommand{\floor}[1]{\left\lfloor #1 \right\rfloor}
\newcommand{\ceil}[1]{\left\lceil #1 \right\rceil}

\renewcommand{\th}{\ifmmode{^{\textrm{th}}}\else{\textsuperscript{th}\ }\fi}

\newcommand{\norm}[1]{\left\lVert #1 \right\rVert}
\newcommand{\rank}{\operatorname{rank}}
\newcommand{\trace}{\operatorname{tr}}
\newcommand{\transpose}{^{\mathsf{T}}}                                    

\newcommand{\image}{\operatorname{im}}

\newcommand{\prob}[1]{\operatorname{\mathbb{P}}\left[\,#1\,\right]}               
\newcommand{\probover}[2]{\operatorname{\mathbb{P}}_{#1}\left[\,#2\,\right]}      
\newcommand{\expect}[1]{\operatorname{\mathbb{E}}\left[\,#1\,\right]}              
\newcommand{\expectover}[2]{\operatorname{\mathbb{E}}_{#1}\left[\,#2\,\right]}      

\newcommand{\bC}{\mathbb{C}}
\newcommand{\bR}{\mathbb{R}}
\newcommand{\bRnneg}{\mathbb{R}_+}
\newcommand{\bZnneg}{\mathbb{Z}_+}
\newcommand{\cB}{\mathcal{B}}
\newcommand{\cD}{\mathcal{D}}
\newcommand{\cE}{\mathcal{E}}

\newcommand{\lmax}{\lambda_\mathrm{max}}
\newcommand{\lmin}{\lambda_\mathrm{min}}
\newcommand{\mat}{\mathbf{M}}

\newcommand{\connec}{k}
\newcommand{\conduc}{\kappa}

\newcommand{\Symraw}{\mathbb{S}}

\newcommand{\Pd}[1][]{\Symraw_{++}^{\ifthenelse{\equal{#1}{}}{n}{#1}}}

\newcommand{\Psd}[1][]{\Symraw_+^{\ifthenelse{\equal{#1}{}}{n}{#1}}}

\newcommand{\Sym}[1][]{\Symraw^{\ifthenelse{\equal{#1}{}}{n}{#1}}}

\newcommand{\Diag}[1][]{\mathbb{D}^{\ifthenelse{\equal{#1}{}}{n}{#1}}}

\newcommand{\stable}{\operatorname{st.rank}}
\newcommand{\LM}{\operatorname{LM}}
\newcommand{\AGM}{\operatorname{BM}}

\title{Pipage Rounding, Pessimistic Estimators \\ and Matrix Concentration}
\date{}
\numberofauthors{2}
\author{
\alignauthor
Nicholas J.~A.~Harvey\\
       \affaddr{Department of Computer Science}\\
       \affaddr{UBC}\\
       \affaddr{Vancouver, Canada}\\
       \email{nickhar@cs.ubc.ca}
\alignauthor
Neil Olver\\
\affaddr{Department of Mathematics}\\
       \affaddr{MIT}\\
       \affaddr{Cambridge, USA}\\
       \email{olver@math.mit.edu}
       \alignauthor
}
\date{\today}

\begin{document}

\tcversion{}{\pagestyle{empty}}

\maketitle

\begin{abstract}
    Pipage rounding is a dependent random sampling technique that has
several interesting properties and diverse applications.
One property that has been particularly useful is negative correlation of the resulting vector.
Unfortunately negative correlation has its limitations, and there are some further desirable
properties that do not seem to follow from existing techniques.
In particular, recent concentration results for sums of independent random matrices
are not known to extend to a negatively dependent setting.

We introduce a simple but useful technique called \emph{concavity of pessimistic estimators}.
This technique allows us to show concentration of submodular functions and concentration
of matrix sums under pipage rounding.
The former result answers a question of Chekuri et al.\ (2009).
To prove the latter result, we derive a new variant of Lieb's celebrated concavity theorem in matrix
analysis.

We provide numerous applications of these results.
One is to \emph{spectrally-thin trees}, a spectral analog of the \emph{thin trees} that played a crucial
role in the recent breakthrough on the asymmetric traveling salesman problem.
We show a polynomial time algorithm that, given a graph where every edge has effective conductance at least $\conduc$, returns an $O(\conduc^{-1} \cdot \log n / \log \log n)$-spectrally-thin tree.
There are further applications to rounding of semidefinite programs,
to the column subset selection problem,
and to a geometric question of extracting a nearly-orthonormal basis
from an isotropic distribution.

\end{abstract}

\tcversion{}{\newpage \pagestyle{plain}\setcounter{page}{1}}

\section{Introduction}

\textit{Rounding} is a crucial step in the design of many approximation algorithms.
Given a fractional vector satisfying some constraints,
a rounding method produces an integer vector that satisfies those constraints, either exactly or
approximately.
\textit{Randomized rounding} \cite{RT} \cite[Chapter 5]{SW},
in which the coordinates of the fractional vector are rounded randomly and independently,
produces good integer vectors for many applications.
\textit{Dependent rounding} methods, in which the resulting integer vector does not have independent
coordinates, are important in many scenarios where naive randomized rounding does poorly.
Various techniques exist for designing dependent rounding methods
(see, e.g., the surveys \cite{SriSurvey,BansalSlides}).

It is common for a rounding scenario to involve two types of constraints:
hard constraints, which must be satisfied exactly by the integer solution,
and soft constraints, which must be approximately satisfied by the integer solution.
Low-congestion multi-path routing \cite{Sri01},
max cut with given sizes of parts \cite{AS},
thin spanning trees \cite{AGMOS}, and
submodular maximization under a matroid constraint \cite{CCPV,CVZFOCS} are 
examples of problems whose solutions involve such a rounding scenario.
The hard constraint is often membership in an integer polytope
that is defined using combinatorial objects (e.g., matchings or matroids).
The soft constraints are usually simple linear inequalities.

With randomized rounding, the independent choices lead to
concentration of measure phenomena that are useful for handling soft constraints.
For example, Chernoff bounds are commonly used to show that linear inequalities
are approximately satisfied \cite{RT}.
The past decade has seen various uses of \emph{matrix} concentration bounds
(e.g., \cite{AW,RV07,Tropp11})
to show that linear matrix inequalities are approximately satisfied 
by random sampling or rounding.
Such uses have occurred in many diverse areas:
    graph sparsification \cite{SS08}, 
    compressed sensing \cite{VershyninSurvey},
    statistics \cite{TroppMasked},
    machine learning \cite{Recht} and 
    numerical linear algebra \cite{Mahoney}.

With dependent rounding, concentration phenomena can also occur.
\emph{Pipage rounding}, \emph{swap rounding} and \emph{maximum entropy sampling}
are dependent rounding techniques
that have seen many important uses over the past decade
\cite{Sri01,AS,GandhiKPS06,CCPV,AGMOS,CVZFOCS}.
An important feature in some scenarios is that
any Chernoff bound that is valid under independent randomized rounding
remains valid under these dependent rounding techniques.
This fact is proven by showing that the rounded solution has a
\emph{negatively correlated distribution}, then appealing to the fact
that Chernoff bounds remain valid under such distributions~\cite{PS97}.
Unfortunately, commutativity plays a key role in proving that fact,
and these arguments do not seem to extend to matrix concentration bounds,
e.g., \cite{AW,Oliviera,RV07,Tropp11}.
Consequently, these matrix inequalities have so far not been combined with dependent rounding.

We prove the first result showing that matrix concentration bounds
are usable in a dependent rounding scenario.
Our technique is not based on negative correlation, but rather the
fortuitous interaction between pipage rounding and various pessimistic estimators.
In particular, we show that Tropp's matrix Chernoff bound \cite{Tropp11}
has a pessimistic estimator that decreases monotonically under pipage rounding.
As a consequence, we can extend the reach of pipage rounding
from soft constraints that are linear inequalities to soft constraints that are
linear matrix inequalities.
Our proof uses non-trivial techniques from matrix analysis and complex analysis;
in particular, we prove a new variant of Lieb's concavity theorem.

\subsection{Motivation and Results}

One key area where our techniques yield new results is for \emph{thin spanning trees}.
These are intriguing objects in graph theory that relate to
foundational topics, such as nowhere-zero flows \cite{Goddyn},
and the asymmetric traveling salesman problem \cite{AGMOS}.
Given a graph $G$ on $n$ nodes, a spanning tree $T$ of $G$ is \emph{$\alpha$-thin} if, for every
cut, the number of
edges of $T$ crossing the cut is at most $\alpha$ times the number of edges of $G$ crossing
the cut.
It has been conjectured that any graph with connectivity $\connec$ has an $f(\connec)$-thin
spanning tree where $f(\connec) = O(1/\connec)$.
This would imply a constant factor approximation algorithm for the asymmetric traveling salesman
problem~\cite{OS}. 
Asadpour et al.~\cite{AGMOS} give a randomized algorithm to find a spanning tree that is $O(\frac{\log n}{\connec \log \log n})$-thin. 
Later Chekuri et al.~\cite{CVZArxiv,CVZFOCS} gave a simpler algorithm using randomized pipage rounding
or swap rounding.

A \emph{spectrally-thin} spanning tree is a stronger notion that is naturally motivated by work on
spectral sparsification \cite{SS08,BSS09}. 
A spanning tree $T$ is $\alpha$-spectrally-thin if $L_T \preceq \alpha L_G$, where $L_G$ refers to the Laplacian of $G$, and $\preceq$ to the L\"owner ordering of Hermitian matrices.
In \Section{thin}, we show a result on spectrally thin trees that strongly mirrors the result of
Asadpour et al.

\begin{theorem}\TheoremName{specthin}
There is a deterministic, polynomial-time algorithm that given any graph on $n$ nodes where every edge has effective conductance at least $\conduc$, constructs a 
$O(\frac{\log n}{\conduc \log \log n})$-spectrally-thin spanning subtree.
\end{theorem}

This spectral notion of thinness seems to be an important one, as the recent breakthrough of Marcus et al.~\cite{MSS} implies that $O(1/\kappa)$-spectrally-thin trees exist. Details of this connection are given in \Appendix{kadison}. 
It is unknown if similar techniques can show that $O(1/k)$-thin trees exist. 
The best known algorithmic construction of spectrally-thin trees is still \Theorem{specthin}.

This result is a special case of a result in a more abstract geometric setting.
Suppose $V = \set{v_1,\ldots,v_m}$ are unit vectors in $\ell_2^n$ for which
$\sum_{i=1}^m v_i v_i \transpose$ is a multiple of the identity.
Does there exist a subset $V_B = \setst{ v_i }{ i \in B }$
that is a \emph{basis} of $\bR^n$ 
and for which the maximum eigenvalue of $\sum_{i \in B} v_i v_i \transpose$ is small?
The maximum eigenvalue is 1 if and only if $V_B$ is orthonormal,
but an arbitrary $V$ need not contain an orthonormal basis.
Again, the breakthrough of Marcus et al.~\cite{MSS} yields a non-constructive proof
of a basis with maximum eigenvalue $O(1)$; see \Appendix{kadison}.
In \Section{isotropic}, we show how to find in polynomial time a basis $V_B \subseteq V$
for which the maximum eigenvalue of $\sum_{i \in B} v_i v_i \transpose$ is $O(\log n / \log \log n)$.
Previous constructive techniques \cite{AW,Oliviera,RV07,Tropp11} only provide a bound of $O(\log n)$.

Our geometric result also relates to the \newterm{column subset selection} problem in numerical linear algebra \cite{BMD,Tropp09,BDM,DR} 
which seeks to ``approximate'' a matrix $A$ by a small subset of its columns,
under various notions of approximation.
Define the \newterm{stable rank} of $A$ to be 
the Frobenius norm divided by the spectral norm, all squared;
this roughly captures the rank of $A$, ignoring negligibly small singular values.
In numerical linear algebra \cite{BMD,BDM,DR}, the number of columns chosen
is typically much larger than the stable rank.
The operator theory community considers similar questions \cite{BT87,BT91,SSRI,Tropp09},
although the number of columns selected is typically much smaller than the stable rank.
In \Section{CSS}, 
we show that one can efficiently select a \emph{linearly independent} set of columns of size \emph{equal}
to the stable rank, while carefully controlling the maximum singular value.

\subsection{Techniques}

Our results are based on the pipage rounding technique \cite{AS,Sri01,GandhiKPS06,CCPV},
which has had several interesting uses in the recent literature.
Deterministic and randomized forms of pipage rounding exist;
our result applies to both of those, as well as to swap rounding.
Typical uses of pipage rounding involve some of the following ideas.
\begin{itemize}
\item There are processes that iteratively move a point in a matroid base polytope towards
an extreme point, while modifying only two coordinates at a time.
The exchange properties of matroid bases ensure that this is possible.

\item One can define a ``potential function'' on the matroid base polytope
(e.g., the ad hoc functions defined in \cite{AS},
or the multilinear extension of a submodular function \cite{CCPV})
such that the function is concave or convex in directions that increase one coordinate and decrease
another.

\item The randomized form of pipage rounding \cite{Sri01,GandhiKPS06,CVZFOCS}
outputs a matroid base whose elements are negatively correlated
(more precisely, \textit{negative cylinder dependent}).
This ensures that linear functions of that base satisfy the same
Chernoff-type concentration bounds that are satisfied under independent rounding.
\end{itemize}

Our aim is to show that, for various concentration bounds,
the final extreme point satisfies the same bounds that
would be achieved by independent randomized rounding.
For Chernoff bounds this follows from negative correlation,
but for other bounds such a result was not previously known.
\begin{itemize}
\item Let $f$ be a monotone submodular function defined on the ground set of the matroid.
When using randomized pipage rounding, does
the value of $f$ at the final extreme point satisfy the same lower tail bound
as when using independent rounding?
Chekuri et al.~\cite{CVZArxiv} conjectured this to be true,
and they proved such a result when using swap rounding.

\item Let $f$ be a linear function 
mapping points in the matroid base polytope to symmetric matrices.
When using pipage rounding, can the value of $f$ at the final extreme point 
be guaranteed to satisfy the same eigenvalue bounds
as when using independent rounding?
\end{itemize}
It does not seem easy to answer these questions using negative correlation properties.

We present a new approach that leads to a positive answer to both of these questions.
In both cases, we can define a pessimistic estimator \cite{R}
that bounds the probability that randomized rounding
fails to achieve the desired concentration.
We show that these pessimistic estimators are concave 
when one element's sampling probability is increased
and another's is decreased by the same amount.
Due to that concavity property, the base output by randomized pipage rounding
satisfies the same concentration
bounds that would be satisfied under independent randomized rounding.
For the second question (matrix concentration), the pessimistic estimator
can be efficiently evaluated, so deterministic pipage rounding can also be used.

The concavity property of our pessimistic estimator 
for matrix concentration is a non-trivial fact.
We establish that fact by proving a new variant of Lieb's concavity theorem \cite{Lieb},
which is a ``masterpiece of matrix analysis''~\cite{BS}
with deep applications in mathematical physics and quantum information theory
\cite{Carlen,Effros,NC}.
Although there is much interest in the mathematical physics community on
extensions and variants of Lieb's theorem, 
our particular variant does not seem to appear in the literature.

\section{Preliminaries}
\SectionName{prelim}

Let $[m] = \set{1,\ldots,m}$.
For a set $S \subseteq [m]$, the vector $\chi(S) \in \bR^m$ is the characteristic vector of $S$.
For a vector $x \in \bR^m$ and a set $S \subseteq [m]$, 
the notation $x(S)$ denotes $\sum_{i \in S} x_i$.
The vector $e_i$ denotes the $i\th$ standard basis vector of the
finite dimensional vector space that is apparent from context.
The vector $\vec{1}$ denotes a vector whose components are all ones
and whose dimension is apparent from context.
We will use $\bRnneg$ and $\bZnneg$ to denote the nonnegative and positive reals respectively.

Let $\Sym$ denote the space of symmetric, real matrices of size $n \times n$.
Let $\Psd, \Pd \subset \Sym$ respectively denote the cones of positive semidefinite
and positive definite matrices.
Let $\Diag \subseteq \Sym$ denote the space of $n \times n$ diagonal matrices.
Let $\preceq$ denote the L\"owner partial order on symmetric matrices,
i.e., $A \preceq B$ iff $B-A \in \Psd$.
Similarly, $A \prec B$ iff $B-A \in \Pd$.
For $A \in \Sym$, let $\lmax(A)$ and $\lmin(A)$ respectively denote the largest and smallest
eigenvalues of $A$.
For $B \in \Sym$, let $B^+$ denote its Moore-Penrose pseudoinverse.
For $B \in \Psd$, let $B^{+/2} \in \Psd$ denote the positive semidefinite square root of $B^+$.
The image of $B$ is $\image B$
and the orthogonal projection onto $\image B$ is $I_{\image B}$.

The notation $\norm{\cdot}$ denotes the $\ell_2$ norm
for vectors and the $\ell_2$ operator norm for matrices.

If $\cD$ is a distribution, $X \sim \cD$ means that
the random variable $X$ has distribution $\cD$.

\section{Concavity of Pessimistic \tcskip Estimators}

In this section we state the known results on pipage rounding
and our concavity of pessimistic estimators technique.
We then apply this technique in three scenarios, of increasing difficulty:
(1) Chernoff bounds, (2) submodular functions, and (3) matrix concentration.
The latter two results are new,
and in particular are not known to follow using negative correlation.
This pessimistic estimator for matrix concentration
underlies all applications in \Section{applications}.

\subsection{Pipage Rounding}

Pipage rounding is a dependent rounding process
originating in works of Ageev, Srinivasan and Sviridenko \cite{AS,Sri01}.
Calinescu et al.\ \cite{CCPV} generalized it to a matroid setting.
We now state the main results of randomized and deterministic pipage rounding;
a proof sketch is given in \Appendix{pip}.

Let $\mat$ be a matroid on $[m]$ and let $P \subset \bR^m$ be its base polytope.
For all algorithmic applications in this paper, $\mat$ can be presented to the
algorithm via an independence oracle.
A function $g : P \rightarrow \bR$ is said to be \textit{concave under swaps} if
\begin{equation}
\EquationName{fconcave}
\forall p \in P, ~\forall a,b \in [m],
\quad
z \mapsto g\big(p + z (e_a\kern-1pt - e_b)\kern-1pt\big)
~~\text{is concave}.
\end{equation}

\begin{theorem}[Randomized Pipage Rounding] \iftc{\mbox{}\\}
\TheoremName{randpip}
There is a randomized, polynomial-time algorithm that,
given $x \in P$, outputs an extreme point $\hat{x}$ of $P$
with $\expect{\hat{x}}=x$ and such that,
for any $g$ concave under swaps, $\expect{g(\hat{x})} \leq g(x)$.
\end{theorem}

\begin{theorem}[Deterministic Pipage Rounding] \iftc{\mbox{}\\}
\TheoremName{detpip}
There is a deterministic, polynomial-time algorithm that, given
$x \in P$ and a value oracle for a function $g$ that is concave under swaps,
outputs an extreme point $\hat{x}$ of $P$ with $g(\hat{x}) \leq g(x)$.
\end{theorem}

The swap rounding procedure of Chekuri et al.\ \cite{CVZArxiv,CVZFOCS}
also proves \Theorem{randpip} and \Theorem{detpip}.

For $x \in [0,1]^m$, let $\cD(x)$ be the product distribution on $\set{0,1}^m$ with marginals given by $x$,
i.e., $\probover{X \sim \cD(x)}{X_i=1} = x_i$.
Let $\cE \subseteq \set{0,1}^m$.
A \newterm{pessimistic estimator} \cite{R,SrinivasanNotes} for $\cE$
is a function $g : [0,1]^m \rightarrow \bR$
that
satisfies
\tcversion{\begin{gather}
\EquationName{PE}
\probover{X \sim \cD(x)}{X \in \cE} ~\leq~ g(x) \qquad\forall x \in [0,1]^m
    \\\nonumber
\min \{ g(x \!-\! x_i e_i), g(x \!+\!(1\!-\!x_i)e_i) \} \leq g(x)
~~\forall x \!\in\! [0,1]^m , i \!\in\! [m].
\end{gather}}
{\begin{gather}
\EquationName{PE}
\probover{X \sim \cD(x)}{X \in \cE} ~\leq~ g(x) \quad\quad\qquad\forall x \in [0,1]^m
    \\\nonumber
\min \big\{\: g(x - x_i e_i) ,\, g\big(x +(1 - x_i)e_i\big) \:\big\} ~\leq~ g(x)
~~~\forall x \in [0,1]^m,~ i \in [m].
\end{gather}
}
For uses of pessimistic estimators in derandomization,
the function $g$ is also required to be efficiently computable.
That is not required with their use in randomized pipage rounding
as $g$ is not even provided as input to the algorithm.

\begin{claim}[Concavity of Pessimistic Estimators]\iftc{\mbox{}\\}
\ClaimName{concPE}
Let $\cE \subseteq \set{0,1}^m$ and let $g$ be a function 
that satisfies \eqref{eq:PE} and is concave under swaps.

Suppose randomized pipage rounding is started at an initial point
$x_0 \in P$, and let $\hat{x}$ be the (random) extreme point of $P$ that is output.
If $g(x_0) \leq \epsilon$ then $\prob{ \hat{x} \in \cE } \leq \epsilon$.

Suppose deterministic pipage rounding 
is given oracle access to $g$ and an initial point $x_0 \in P$ with $g(x_0) < 1$.
Then the extreme point $\hat{x}$ of $P$ that is output satisfies $\hat{x} \not \in \cE$.
\end{claim}

We omit the proof of \Claim{concPE} as it is an easy consequence of
\Theorem{randpip} and \Theorem{detpip}.

\subsection{Chernoff bound}

Let us start with a simple result to illustrate the technique.
First we state the Chernoff bound in convenient notation.
We discuss only the right tail; an analogous result holds for the left tail.
Fix any vector $w \in [0,1]^m$.
For $t \in \bR$ and $\theta > 0$, define $g_{t,\theta} : [0,1]^m \rightarrow \bR$ by
$$
g_{t,\theta}(x) ~:=~
    e^{-\theta t} \cdot \expectover{X \sim \cD(x)}{ e^{\theta w \transpose X}}.
$$
Let $\mu = w \transpose x$ and $\delta \geq 0$. Then
\tcversion{
\begin{align}\nonumber
\probover{X \sim \cD(x)}{ w\transpose X \geq t }
    &~\leq~ \inf_{\theta>0} \, g_{t,\theta}(x)
    \\
    \EquationName{chernoffRHS}
    \text{and}\tcversion{\quad}{\qquad\qquad}
g_{(1+\delta) \mu ,\, \ln(1+\delta)}(x)
    &~\leq~ \Big( \frac{e^\delta}{(1+\delta)^{1+\delta}} \Big)^\mu.
\end{align}
}{
\begin{equation}
\EquationName{chernoffRHS}
    \probover{X \sim \cD(x)}{ w\transpose X \geq t } ~\leq~ \inf_{\theta>0} \, g_{t,\theta}(x) \qquad
    \text{and}\qquad g_{(1+\delta) \mu ,\, \ln(1+\delta)}(x) ~\leq~ \Big( \frac{e^\delta}{(1+\delta)^{1+\delta}} \Big)^\mu.
\end{equation}
}

The following claim is proven in \Appendix{concavity}.

\begin{claim}
\ClaimName{chernoffconcave}
$g_{t,\theta}$ is concave under swaps.
\end{claim}

Consequently, \Claim{concPE} implies the following result.

\begin{corollary}
If randomized pipage rounding starts at $x_0 \in P$
and outputs the extreme point $\hat{x}$ of $P$ then,
$\forall w \in [0,1]^m ,\, \delta \geq 0$,
\begin{equation}
\EquationName{pipageChernoffRHS}
\prob{ w \transpose \hat{x} \geq (1+\delta) \mu }
    ~\leq~ \Big( \frac{ e^\delta }{ (1+\delta)^{1+\delta} } \Big)^{\mu}
\end{equation}
where $\mu = w \transpose x_0$.
Furthermore, if this right-hand side is strictly less than $1$,
then deterministic pipage rounding outputs an extreme point $\hat{x}$ of $P$
with $w \transpose \hat{x} < (1+\delta) \mu$.
\end{corollary}

The key point is that the right-hand sides of \eqref{eq:chernoffRHS} and
\eqref{eq:pipageChernoffRHS} are the same.
Chekuri et al.~\cite{CVZArxiv} proved this fact using negative correlation of $\hat{x}$,
generalizing a result of Srinivasan~\cite{Sri01}.

\subsection{Submodular functions}

Chekuri et al.~\cite[Theorem 1.3]{CVZArxiv} prove an analog of the Chernoff bound
for concentration of submodular functions under independent rounding.
They show that the same bound remains true under swap rounding~\cite[Theorem 1.4]{CVZArxiv} and
ask whether it remains true under pipage rounding.

Formally, let $f : \set{0,1}^m \rightarrow \bR$
be a non-negative, monotone, submodular function with marginals in $[0,1]$.
The \textit{multilinear extension} of $f$ is $F : [0,1]^m \rightarrow \bR$ with
$ F(x) := \expectover{X \sim \cD(x)}{f(X)} $.
For $t \in \bR$ and $\theta<0$, define $g_{t,\theta} : [0,1]^m \rightarrow \bR$ by
\[
   g_{t,\theta}(x) ~:=~ e^{-\theta t} \cdot \expectover{X \sim \cD(x)}{ e^{\theta f(X)} }.
   \]
The left tail bound of Chekuri et al.\ is: with $\mu = F(x) ,\, \delta \in [0,1)$,
\tcversion{
\begin{align*}
\probover{X \sim \cD(x)}{ f(X) \leq t }
    &~\leq~ \inf_{\theta<0} \, g_{t,\theta}(x) \\
    \text{and}\tcversion{\quad}{\qquad\qquad}
g_{(1-\delta) \mu ,\, \ln(1-\delta)}(x)
    &~\leq~ \exp( - \delta^2 \mu / 2 ).
\end{align*}
}{
\[
\probover{X \sim \cD(x)}{ f(X) \leq t }
    ~\leq~ \inf_{\theta<0} \, g_{t,\theta}(x) 
    \qquad\text{and}\qquad
g_{(1-\delta) \mu ,\, \ln(1-\delta)}(x)
    ~\leq~ \exp( - \delta^2 \mu / 2 ).
\]
}

The following claim is proven in \Appendix{concavity}.

\begin{claim}
\ClaimName{submodconcave}
$g_{t,\theta}$ is concave under swaps.
\end{claim}

\Claim{concPE} implies the following result,
answering an open question of Chekuri et al.\ \cite[p. 3]{CVZArxiv}.

\begin{corollary}
If randomized pipage rounding starts at $x_0 \in P$
and outputs the extreme point $\hat{x}$ of $P$ then,
letting $\mu = F(x_0)$,
we have 
$ \prob{ f(\hat{x}) \leq (1-\delta) \mu } \leq \exp(-\delta^2 \mu / 2)$.
\end{corollary}

\noindent
Chekuri et al.~\cite[p.~583]{CVZFOCS} state that this fact
does not follow from negative correlation of $\hat{x}$.

\subsection{Matrix Concentration}

\onote{Modified}
Tropp~\cite{Tropp11}, improving on Ahlswede-Winter~\cite{AW} and Oliviera~\cite{Oliviera},  proves a beautiful analog of the Chernoff bound
for sums of independent random matrices.
We state a simplified form here.

\begin{theorem}
\TheoremName{tropp}
Let $M_1,\ldots,M_m \in \Psd$ satisfy $M_i \preceq R \cdot I$.
For $t \in \bR$ and $\theta>0$, define $g_{t,\theta} : [0,1]^m \rightarrow \bR$ by
$$
   g_{t,\theta}(x) ~:=~ e^{-\theta t} \cdot 
      \trace \exp\Big(\sum_{i=1}^m \log \expectover{X \sim \cD(x)}{ e^{\theta X_i M_i} } \Big).
$$
Then, for $\mu \geq \norm{ \expectover{X \sim \cD(x)}{ \sum_i X_i M_i } }$ and $\delta \geq 0$,
\tcversion{
\begin{align*}
\probover{X \sim \cD(x)}{ \norm{\smallsum{i}{} X_i M_i} \geq t }
    &~\leq~ \inf_{\theta>0} \, g_{t,\theta}(x) \\
    \text{and}\tcversion{\quad}{\qquad\qquad\quad}
g_{(1+\delta) \mu ,\, \ln(1+\delta)}(x)
    &~\leq~ n \cdot \Big( \frac{e^\delta}{(1+\delta)^{1+\delta}} \Big)^{\mu/R}.
\end{align*}
}{
$$
\probover{X \sim \cD(x)}{ \norm{\smallsum{i}{} X_i M_i} \geq t }
    ~\leq~ \inf_{\theta>0} \, g_{t,\theta}(x)
    \qquad\text{and}\qquad
g_{(1+\delta) \mu ,\, \ln(1+\delta)}(x)
    ~\leq~ n \cdot \Big( \frac{e^\delta}{(1+\delta)^{1+\delta}} \Big)^{\mu/R}.
$$
}
\end{theorem}

The following is our main lemma on pessimistic estimators.
The proof is in \Appendix{concavity}.

\begin{lemma}
\LemmaName{pipage}
$g_{t,\theta}$ is concave under swaps.
\end{lemma}

Consequently, \Claim{concPE} implies the following result.

\begin{corollary}
\CorollaryName{matrixconcentration}
Let $P$ be a matroid base polytope and let $x_0 \in P$.
Let $M_1,\ldots,M_m \in \Psd$ satisfy $M_i \preceq R \cdot I$.
Let $\mu \geq \norm{\expectover{X \sim \cD(x_0)}{\sum_i X_i M_i}}$.
If randomized pipage rounding starts at $x_0$
and outputs the extreme point $\hat{x}=\chi(S)$ of $P$ then 
we have
\begin{equation}
\EquationName{matrixPE}
\prob{ \norm{\smallsum{i \in S}{} M_i } \geq (1+\delta) \mu }
    \:\leq\: n \cdot \Big( \frac{ e^\delta }{ (1+\delta)^{1+\delta} } \Big)^{\mu/R}.
\end{equation}
Furthermore, if this right-hand side is strictly less than $1$, then
deterministic pipage rounding outputs an extreme point $\hat{x} = \chi(S)$ of $P$ with
$\norm{\sum_{i \in S} M_i} < (1+\delta) \mu$.
\end{corollary}

The inequalities in \Theorem{tropp} involve non-trivial matrix analysis,
such as operator concavity of $\log$ and Lieb's celebrated concavity theorem \cite{Lieb}.
It seems that even those results do not suffice to prove \Lemma{pipage}.
To prove it, we derive a new variant of Lieb's theorem (\Theorem{liebvariant}).
Lieb \cite{Lieb} proved several related concavity theorems;
for us, the most relevant form is:

\begin{theorem}[Lieb \cite{Lieb}]
\TheoremName{lieb}
Let $L, K \in \Sym$ and $C \in \Pd$.
Then $z \mapsto \trace \exp\big( L + \log(C+zK) \big)$ is concave in a neighborhood of\/ $0$.
\end{theorem}

The main technical result of this paper is:

\begin{theorem}
\TheoremName{liebvariant}
\label{LIEBVARIANT}
Let $L \in \Sym$, $C_1, C_2 \in \Pd$ and $K_1, K_2 \in \Psd$.
Then the univariate function
\begin{equation}
\EquationName{exploglog}
z ~\mapsto~ \trace \exp\Big( L + \log( C_1 + z K_1 ) + \log( C_2 - z K_2 ) \Big)
\end{equation}
is concave in a neighborhood of\/ $0$.
\end{theorem}

There are several known approaches to proving Lieb's theorem.
The simplest is Tropp's approach \cite{Tropp12};
however, his proof is based on joint convexity of quantum entropy,
which is itself usually proven using Lieb's theorem.
We were unable to prove \Theorem{liebvariant} using Tropp's approach.
Lieb's original proof~\cite{Lieb},
which proves concavity by directly analyzing the second derivative,
involves numerous delicate steps of matrix analysis.
We were able to adapt this approach to prove a weaker form of \Theorem{liebvariant}
that requires some additional commutativity assumptions; details are in \Appendix{lieb}.
This weaker result suffices to prove \Lemma{pipage}.
Epstein~\cite{Epstein} gives an elegant approach to proving Lieb's theorem
using complex analysis, and in particular powerful results concerning \emph{Herglotz functions}.
Our proof of \Theorem{liebvariant}, which appears in \Appendix{epstein},
is an adaptation of Epstein's approach.

\medskip
\noindent \tcversion{\textsc{Remark.}}{\textbf{Remark.}}
Another well-known matrix concentration inequality is the Ahlswede-Winter \cite{AW} inequality,
for which pessimistic estimators were studied by Wigderson and Xiao \cite{WX}.
It is natural to wonder whether we could have used their pessimistic estimators instead.
Unfortunately they do not seem applicable for our scenario.
The issue is that the Ahlswede-Winter inequality is most effective
for analyzing sums of i.i.d.\ random matrices,
due to some inequalities that arise in their analysis.
In our scenario, due to the way that pipage rounding works,
we require non-i.i.d.\ product distributions,
so it is much more convenient to base our approach on \Theorem{tropp}.

\section{Applications}
\SectionName{applications}

\subsection{Rounding of semidefinite programs}

Let $\mat$ be a matroid and let $P \subset \bR^n$ be its base polytope.
Consider the spectrahedron
\begin{equation}
\EquationName{Qdef}
Q ~:=~ P \:\intersect\: \Bigl\{\,x \in \bR^m \;:\; \sum_{i=1}^m x_i A_i \preceq B \,\Bigr\},
\end{equation}
where each $A_1,\ldots,A_m, B \in \Psd$.
We think of $P$ as specifying ``hard'' constraints and the semidefinite constraint as being ``soft''.

\begin{theorem}
\TheoremName{main}
Suppose that $A_i \preceq B$ for all $i$.
If randomized pipage rounding starts at $x_0 \in Q$
and outputs the extreme point $\chi(S)$ of $P$, then 
$ \prob{ \smallsum{i \in S}{} A_i \preceq \alpha B } \geq 1-1/n $,
for some $\alpha = O( \log n / \log \log n )$.
Furthermore, if deterministic pipage rounding starts at $x_0 \in Q$,
then it outputs an extreme point $\chi(S)$ of $P$ with
$\smallsum{i \in S}{} A_i \preceq \alpha B$.
\end{theorem}

This theorem is optimal with respect to $\alpha$, as discussed below.
The hypothesis that $A_i \preceq B$ is a ``width'' condition
that commonly arises in optimization and rounding.

\begin{proof}
Recall the notation defined in \Section{prelim}.
Let $M_i = B^{+/2} A_i B^{+/2}$.
By standard arguments, 
\begin{align*}
    \smallsum{i=1}{m} x_i A_i \preceq B
    &\quad\iff\quad
    \smallsum{i=1}{m} x_i M_i \preceq I_{\image B} \\
\text{and}\quad
    \smallsum{i \in S}{} A_i \preceq \alpha B
    &\quad\iff\quad
    \smallsum{i \in S}{} M_i \preceq \alpha I_{\image B}.
\end{align*}
We assume that $A_i \preceq B$, so $\lmax(M_i) \leq 1$.
Apply \Corollary{matrixconcentration}
with $\delta = 4 \log n / \log \log n$, $\mu=1$ and $R=1$.
A standard calculation shows that the right-hand side of \eqref{eq:matrixPE} is less than $1/n$.
\end{proof}

Chekuri, Vondr\'ak and Zenklusen \cite{CVZArxiv,CVZFOCS} considered the problem of rounding a point
in a matroid polytope to an extreme point, subject to additional packing constraints.
Their result generalizes the low-congestion multi-path routing problem studied earlier by
Srinivasan et al.~\cite{Sri01,GandhiKPS06},
but it is itself a special case of \Theorem{main}
where the matrices $A_i$ and $B$ are diagonal.
The factor $\alpha = O(\log n / \log \log n)$ is optimal in \Theorem{main}
because it is optimal for rounding this low-congestion multi-path routing problem,
and even for the congestion minimization problem \cite{LRS}.

\subsection{Rounding an isotropic distribution to a \tcskip nearly orthonormal basis}
\SectionName{isotropic}

Let $w_1,\ldots,w_m \in \bR^n$ satisfy $\norm{w_i} = 1$ for all $i$.
Let $p_1,\ldots,p_m$ be a probability distribution on these vectors
such that the covariance matrix is $\sum_i p_i w_i w_i \transpose = I/n$.
A random vector drawn from that distribution is said to be in \newterm{isotropic position}.

\begin{theorem}
\TheoremName{isotropic}
There is a polynomial time algorithm (either randomized or deterministic)
to compute a subset $S \subseteq [m]$ such that $\setst{ w_i }{ i \in S }$ forms a
\emph{basis} of\/ $\bR^n$,
and for which $ \norm{ \sum_{i \in S} w_i w_i \transpose } \leq \alpha $,
where $\alpha =  O(\log n / \log \log n)$.
\end{theorem}

As is discussed in \Appendix{kadison}, the recent breakthrough on the Kadison-Singer problem~\cite{MSS} implies the following existential result:
\begin{theorem}
\TheoremName{strongIsotropic}
There exists $S \subseteq [m]$ such that $\setst{ w_i }{ i \in S }$
forms a basis of\/ $\bR^n$, and for which
$\norm{ \sum_{i \in S} w_i w_i \transpose } = O(1)$.
\end{theorem}

We now prove \Theorem{isotropic} using \Theorem{main}.
Let $\mat$ be the linear matroid corresponding to the vectors $\set{w_1,\ldots,w_m}$.
Let $P$ be the base polytope of that linear matroid.
Let $r : 2^{[m]} \rightarrow \bZnneg$ be the rank function of that matroid,
i.e., 
$r(S) = \operatorname{dim} \big(\operatorname{span} \setst{ w_i }{ i \in S }\big)$.
Then 
\[
P ~:=~ \setst{ x \in \bRnneg^n }{ x(J) \leq r(J) ~\forall J \subseteq [m] ~~\text{and}~~ x([m]) = r([m]) }.
\]

Define $A_i = w_i w_i \transpose$, $B = I$ and
\[ 
Q ~=~ P \:\intersect\: \Bigl\{\, x \in \bR^m \,:\, \sum_i x_i A_i \preceq B \,\Bigr\}.
\]
Let $x = n \cdot p$.
Then the following claim and the hypothesis that $\sum_i p_i w_i w_i \transpose = I/n$
show that $x \in Q$.

\begin{claim}
\ClaimName{pinP}
$x \in P$.
\end{claim}

Since $\norm{w_i}=1$, we have $A_i = w_i w_i \transpose \preceq I = B$.
\Theorem{main} gives an algorithm to construct an
extreme point $\chi(S)$ of $P$ 
for which $\sum_{i \in S} A_i \preceq \alpha \cdot B$,
with $\alpha = O(\log n / \log \log n)$.
Since $P$ is the base polytope of $\mat$,
$\setst{ w_i }{ i \in S }$ forms a basis of $\bR^n$.
Finally, $ \sum_{i \in S} w_i w_i \transpose \preceq \alpha \cdot I $.
This completes the proof of \Theorem{isotropic}, modulo the proof of \Claim{pinP}.

In \Appendix{Psubmodular}, we show that \Theorem{isotropic} can be generalized from a decomposition
of the identity into rank-one matrices $w_i w_i \transpose$ to a decomposition into matrices of
arbitrary rank.
The proof of \Claim{pinP} is analogous to the proof of \Claim{pinPprime}.
We remark that \Theorem{strongIsotropic} is not known to have a generalization
to matrices of arbitrary rank.

\subsection{Thin trees}
\SectionName{thin}
\label{THIN}

Let $G = (V,E)$ be a graph.
For convenience we assume that $V = [n]$.
The \newterm{cut} defined by $U \subseteq V$ is 
\[
\delta_G(U) \:=\: \delta(U) \:=\: \setst{ uv \in E }{ \text{ exactly one of $u$ and $v$ is in $U$}}.
\]
For a subgraph $T$ of $G$, 
let $\delta_T(U)$ denote all edges of $T$ with exactly one endpoint in $U$.

\begin{definition}
A subgraph $T$ of $G$ is called \newterm{$\epsilon$-thin} if
$\card{ \delta_T(U) } \leq \epsilon \cdot \card{ \delta_G(U) }$
for all $U \subseteq V$.
\end{definition}

\begin{conjecture}[Goddyn \cite{Goddyn}]
\ConjectureName{goddyn}
Every graph with connectivity at least $\connec$ has an $f(\connec)$-thin spanning subtree,
for some function $f$ that vanishes as $\connec$ tends to infinity.
\end{conjecture}

The crucial detail in this conjecture is that the function $f$ should not
depend on the size of the graph.
The best progress on this conjecture for general graphs is as follows.

\begin{theorem}[Asadpour et al.~\cite{AGMOS}]
\TheoremName{asadpour}
Let $G$ be a graph with $n$ vertices and connectivity $\connec$.
Then $G$ has a $O\big( \frac{\log n}{\connec \log \log n} \big)$-thin spanning subtree.
Moreover, there is a randomized, polynomial time algorithm to construct such a tree.
\end{theorem}

Now we define spectrally-thin trees and prove an analog of this theorem.
The Laplacian of $G$ is the symmetric matrix $L_G$ with rows and columns indexed by $V$ 
defined by
\[ L_G ~:=~ \sum_{uv \in E} (e_u - e_v) (e_u - e_v) \transpose. \]

\begin{definition}
Let $T$ be a spanning subtree of $G$ and let $L_T$ be the Laplacian of $T$.
The tree $T$ is \newterm{$\epsilon$-spectrally-thin} if
$L_T \preceq \epsilon L_G $.
\end{definition}

Any tree that is $\epsilon$-spectrally-thin is also $\epsilon$-thin,
because
\[
    \card{\delta_T(U)}
    ~=~ \chi(U) \transpose \, L_T \, \chi(U)
    ~\leq~ \epsilon \cdot \chi(U) \transpose \, L_G \, \chi(U)
    ~=~ \epsilon \cdot \card{\delta_G(U)}.
\]
The converse is not true.
Moreover, the connectivity hypothesis in \Theorem{asadpour} does not suffice\footnote{
    This result was independently observed by M.~de Carli Silva, N.~Harvey and C.~Sato,
    and by M. Goemans~\cite{GoemansTalk}, using slightly different examples.
}
to obtain a good spectrally-thin tree.
The proof is in \Appendix{nospectrallythintree}.

\begin{theorem}
\TheoremName{nospectrallythintree}
For every $n, k \geq 1$, there exists a weighted graph with $n$ vertices and connectivity $\connec$
that does not have an $o( \sqrt{n}/\connec )$-spectrally-thin spanning subtree.
\end{theorem}

Nevertheless, if we strengthen the connectivity lower bound to a lower bound
on the effective conductances, then we have the following construction of spectrally-thin trees.
For an edge $e = uv \in E$, the \newterm{effective resistance} in $G$ between $u$ and $v$ is
$R_e := (e_u - e_v)\transpose L_G^+ (e_u - e_v)$.
The \newterm{effective conductance} in $G$ between $u$ and $v$ is $C_e := 1/R_e$.

\begin{theorem}
\TheoremName{conductanceweightedtree}
Let $G$ be a graph with $n$ vertices such that $\conduc \leq C_e$ for every edge $e$.
Then there is a polynomial time algorithm (either randomized or deterministic)
 to construct a $O\big( \frac{\log n}{\conduc \log \log n} \big)$-spectrally-thin
spanning subtree of $G$.
\end{theorem}

\Theorem{conductanceweightedtree} follows directly from \Theorem{main},
letting $\mat$ be the graphic matroid corresponding to $G$.
It also follows from \Theorem{isotropic}, as we show in \Appendix{transitive}.
That viewpoint is advantageous, since \Theorem{strongIsotropic} then immediately implies 
\begin{theorem}
\TheoremName{verythintree}
Let $G$ be a graph with $n$ vertices such that $\conduc \leq C_e$ for every edge $e$.
Then $G$ has a $O( 1/\conduc )$-spectrally-thin spanning subtree.
\end{theorem}
We are not aware of any formal connection between \Theorem{verythintree}
and \Conjecture{goddyn} or the traveling salesman problem.

Although \Theorem{asadpour} and \Theorem{conductanceweightedtree} are formally incomparable,
it is worth understanding their similarities and differences.
Both results have a seemingly suboptimal factor of $\log n / \log \log n$.
\Theorem{asadpour} requires only a connectivity lower bound,
which is important in applications \cite{Goddyn,AGMOS},
but the resulting tree is thin, not spectrally-thin; also, their algorithm is randomized.
\Theorem{conductanceweightedtree} requires a conductance lower bound
(which is stronger than a connectivity lower bound),
but the resulting tree is spectrally-thin (which is stronger than being thin);
also, our algorithm can be made deterministic.
The use of randomization seems quite inherent in the algorithms \cite{AGMOS,CVZFOCS}
for \Theorem{asadpour}, as the thinness condition involves controlling exponentially many cuts,
which seems difficult to accomplish by a deterministic, polynomial-time algorithm.

The quantities $\connec$ and $\conduc$ can be related in certain classes of graphs.
We say that a family of graphs has \newterm{nearly equal resistances} if 
there is a constant $c$ (independent of the number of vertices)
such that $R_e \leq c R_f$ for all edges $e,f$.
For example, any Ramanujan graph has nearly equal resistances.
Edge-transitive graphs, such as hypercubes,
have nearly equal (in fact, exactly equal) resistances.

\begin{corollary}
\CorollaryName{transitive}
Let $G$ be a graph with $n$ vertices, nearly equal resistances,
and connectivity $\connec$.
Then there is a deterministic, polynomial time algorithm to construct a $O\big( \frac{\log n}{\connec \log \log n} \big)$-spectrally-thin tree of $G$.
\end{corollary}

The proof is in \Appendix{transitive}.

\subsection{Column-subset selection}
\SectionName{CSS}

\todo{Consequences of MSS here?}

Column-subset selection is an important topic in
numerical linear algebra \cite{BMD,Tropp09,DR,BDM}.
Similar questions are considered in operator theory \cite{BT87,BT91,SSRI,Tropp09,Youssef}.
In this section we prove a non-isotropic analog of \Theorem{isotropic},
which gives a new result on column-subset selection.
For a real matrix $A$, let $\norm{A}_F = \sqrt{ \trace A \transpose A }$ denote its Frobenius norm.
The \newterm{stable rank} of $A$ is $\stable(A) := \norm{A}_F^2/\norm{A}^2$.

\begin{theorem}
\TheoremName{kashin}
Let $A$ be a real matrix of size $n \times m$
whose columns are denoted $a_1,\ldots,a_m$.
Suppose that $\norm{a_i} = 1 ~\,\forall i$.
Then there is a deterministic, polynomial time algorithm to compute $S \subseteq [m]$
of size $\card{S} \geq \floor{ \stable(A) }$
such that $\setst{ a_i }{ i \in S }$ is linearly independent,
and $\norm{ \sum_{i \in S} a_i a_i \transpose } \leq O(\log n / \log \log n)$.
\end{theorem}

This is optimal with respect to $\card{S}$ as it can happen that $\stable(A)=\rank(A)$,
in which case $\setst{ a_i }{ i \in S }$ is linearly dependent whenever $\card{S} > \stable(A)$.

We now prove \Theorem{kashin} using \Theorem{main}.
Note that $\norm{A}_F^2=m$.
Let $p \in \bR^m$ be the vector with $p_i = \floor{m/\norm{A}^2}/m$ for all $i$.
Note that $\sum_i p_i = \floor{\stable(A)}$.
We claim that $p$ can be viewed as a ``fractional set'' of linearly independent vectors
of size $\floor{\stable(A)}$.
Formally, for any set $T \subseteq [m]$, let $A_T$ denote the submatrix of $A$
consisting of the columns in $T$.
Define the following family of sets
\[
    \cB
    ~=~ \setst{ I \subseteq [m] }{ \rank A_I=\card{I}=\floor{\stable(A)} }.
\]
Then $\cB$ is the base family of the linear matroid corresponding to $A$,
truncated to rank $\floor{\stable(A)}$.
Let $\mat$ denote that matroid and let $P$ denote its base polytope.

\begin{claim}
\ClaimName{CSpinP}
$p \in P$.
\end{claim}

The proof is in \Appendix{CSS}.
Given this claim, all that remains is an easy application of \Theorem{main}.
Define $A_i = a_i a_i \transpose$, $B = I$ and
\[
Q ~=~ P \:\intersect\: \Bigl\{\: x \in \bR^m \;:\; \sum_i x_i A_i \preceq B \:\Bigr\}.
\]
We have $p \in Q$ by \Claim{CSpinP} and the fact that
\[
\sum_i p_i A_i 
~\preceq~ \sum_i \frac{A_i}{\norm{A}^2}
~=~ \frac{A A \transpose}{\norm{A}^2}
~\preceq~ I ~=~ B.
\]
Note that $A_i = a_i a_i \transpose \preceq I = B$.
\Theorem{main} gives a deterministic algorithm 
to construct an extreme point $\chi(S)$ of $P$
for which $\sum_{i \in S} A_i \preceq \alpha \cdot B$,
with $\alpha = O(\log n / \log \log n)$.
Since $S$ is a base of $\mat$,
the set $\setst{ a_i }{ i \in S }$ has rank equal to $\card{S} = \floor{\stable(A)}$.
This completes the proof of \Theorem{kashin}.

\paragraph*{Acknowledgements.}
N.~Harvey thanks Joel Friedman and Mohit Singh for numerous enlightening discussions.
We thank Isaac Fung for collaborating at a preliminary stage of this work.
We also thank Christos Boutsidis,
Joseph Cheriyan, Satoru Fujishige, Michel Goemans, Mary Beth Ruskai, Nikhil Srivastava,
Joel Tropp, Roman Vershynin and Jan Vondr\'ak for helpful discussions and suggestions.

\clearpage
\appendix
\section{Pipage Rounding}
\AppendixName{pip}

\onote{Changed $f$ to $g$.}

Let $p$ be a point in the matroid polytope $P$ and assume that $g$ satisfies \eqref{eq:fconcave}.
Delete all coordinates of $p$ that are equal to zero and consider the residual problem.
It is well-known that, for any such point $p$, there exists a chain of sets
$\emptyset = C_0 \subseteq C_1 \subseteq \cdots C_k \subseteq [m]$
whose corresponding constraints of $P$ span the constraints that are tight at $p$.
If $\card{C_i \setminus C_{i-1}}=1$ for every $i$ then these give $m$ linearly independent tight
constraints, so the point $p$ is an extreme point.
Otherwise there is some set $C_i$, $i \geq 1$, for which $\card{C_i \setminus C_{i-1}} > 1$.
In this case $p$ is not an extreme point.
To see this, let $a$ and $b$ be distinct elements of $C_i \setminus C_{i-1}$.
Note that the point $p + z(e_a - e_b)$ satisfies all the constraints that are tight at $p$.
So, for all $z$ in some open neighborhood of $0$,
the point $p + z(e_a - e_b)$ is still feasible for $P$.

Define
\begin{align*}
    \ell &= \min \setst{ z \in \bR }{ p + z(e_a - e_b) \in P }\\
    \text{and}\qquad u &= \max \setst{ z \in \bR }{ p + z(e_a - e_b) \in P }.
\end{align*}
Define
\[ 
p^\ell = p + \ell(e_a - e_b)
\qquad\text{and}\qquad
p^u = p + u(e_a - e_b).
\]
Since $g\big(p+z(e_a-e_b)\big)$ is concave, we must have either
\[
g\big(p^\ell) \leq g(p)
\qquad\text{or}\qquad
g\big(p^u\big) \leq g(p).
\]
Furthermore, both $p^\ell$ and $p^u$ lie on a lower-dimensional face than $p$ does.
So starting from some initial $p^0 \in P$, $m$ iterations suffice to find an extreme point $\hat{p}$ of $P$
with $g(\hat{p}) \leq g(p^0)$. 

The randomized version of pipage rounding does not even need access to the function $g$.
Instead, it simply chooses the next point $p'$ to be $p^{\ell}$ with probability $\frac{u}{u-\ell}$, or $p^{u}$ with probability $\frac{-\ell}{u-\ell}$.
This ensures that $\expect{p'} = p$, and the concavity of $g$ yields $\expect{g(p')} \leq g(p)$.
Thus applying this procedure to some initial point $p_0 \in P$ until an extreme point $\hat{p}$ is obtained, $\hat{p}$ satisfies $\expect{\hat{p}} = p_0$ and $\expect{g(\hat{p})} \leq g(p^0)$.

\section{Proofs of concavity under swaps}
\AppendixName{concavity}

\begin{proofof}{\Claim{chernoffconcave}}
We can rewrite
$$
g(x) ~=~ e^{-\theta t} \cdot \smallprod{i}{} \big( 1 + x_i (e^{\theta w_i}\!-\!1) \big).
$$
Rewriting $g\big(x+z(e_a-e_b)\big)$ in this way, all factors are non-negative and
only two of them depend on $z$, so for some $c \geq 0$ 
\begin{align*}
&\frac{d^2}{dz^2} g\big(x+z(e_a-e_b)\big) \\
&\:=\: c \cdot \frac{d^2}{dz^2} \Big(
        \big( 1 + (x_a\!+\!z) (e^{\theta w_a}\!-\!1) \big)
        \big( 1 + (x_b\!-\!z) (e^{\theta w_b}\!-\!1) \big)
    \Big) \\
&\:=\: c \cdot \Big( -2 \big(e^{\theta w_a}-1 \big)\big(e^{\theta w_b}-1 \big) \Big).
\end{align*}
This is non-positive so $g$ is concave under swaps.
\end{proofof}

\begin{proofof}{\Claim{submodconcave}}
Recall that $\theta < 0$.
Define $h : \set{0,1}^m \rightarrow \bR$ by $h(X) = e^{\theta f(X)}$.
By \Claim{supermodular}, $h$ is a supermodular function.
Its multilinear extension is
$$
H(x) ~=~ \expectover{X \sim \cD(x)}{h(X)}.
$$
\onote{Modified; this is fine right, don't need $e^{\theta f(X)}$?}
Since $-h$ is submodular, it follows from results of Calinescu et al.~\cite{CCPV} that
$\frac{\partial^2 H}{\partial x_i \partial x_j} \geq 0$ for any $i,j \in [m]$. \onote{The referee indicated that this sign was wrong, but it does seem to be correct.}
Since $g(x) = e^{-\theta t} \cdot H(x)$, 
the second derivative of
$$ z ~\mapsto~ g\big(x + z(e_i-e_j)\big) $$ is non-positive.
Thus $g$ is concave under swaps.
\end{proofof}

\begin{claim}
\ClaimName{supermodular}
Let $f : 2^{[m]} \rightarrow \bR$ be non-decreasing and submodular.
Let $g : \bR \rightarrow \bR$ be non-increasing and convex.
Then $g \circ f$ is supermodular. 
\end{claim}
\begin{proof}
We require the following property of convex functions. Suppose $a,b,c,d$ satisfy
\begin{equation}
\EquationName{abcd}
   a ~\leq~ \min \set{b,c} ~\leq~ \max \set{b,c} ~\leq~ d.
\end{equation}
Then any function $g$ that is convex on $[a,d]$ satisfies
\begin{equation}
\EquationName{abcd2}
   \frac{ g(d)-g(c) }{ d-c } ~\geq~ \frac{ g(b)-g(a) }{ b-a }.
\end{equation}
Fix any $A \subseteq B \subseteq [m]$, and an element $x \in [m] \setminus B$.
Define
\begin{align*}
   a &:= f(A), \quad b := f(A+x), \quad c := f(B), \\
   d &:= f(B) + f(A+x) - f(A), \quad e := f(B+x).
\end{align*}
Since $f$ is non-decreasing, \eqref{eq:abcd} holds.
Since $f$ is submodular, $e \leq d$ holds.
Since $g$ is non-increasing, $g(e) \geq g(d)$.
Combining that with \eqref{eq:abcd2} and the observation that $d-c = b-a$,
we obtain
$$
   g(e)-g(c) ~\geq~ g(d)-g(c) ~\geq~ g(b)-g(a)
$$
That is,
$$
    g\big(f(B+x)\big) - g\big(f(B)\big) ~\geq~ g\big(f(A+x)\big) - g\big(f(A)\big),
$$
so $g \circ f$ is supermodular.
\end{proof}

\begin{proofof}{\Lemma{pipage}}
We will show that
$$
\forall x \in (0,1)^m, ~\forall a,b \in [m],
~~~
z \mapsto g_{t,\theta}\big(x + z (e_a\! - e_b)\!\big)
~~\text{is concave}.
$$
The boundary of $[0,1]^m$ is handled by continuity.
Note that
\[
\expectover{X \sim \cD(x)}{ e^{\theta X_i M_i} }
~=~ x_i \cdot e^{\theta M_i} \,+\, (1-x_i) \cdot I
~=:~ C_i.
\]
Adding $z$ (sufficiently small) to the sampling probability of coordinate $i$,
the expectation becomes
\begin{align*}
\expectover{X \sim \cD(x + z e_i)}{ e^{\theta X_i M_i} }
~&=~ (x_i+z) \cdot e^{\theta M_i} \,+\, (1-x_i-z) \cdot I \\
~&=~ C_i + z \underbrace{\Big(\: e^{\theta M_i} \:-\: I \:\Big)}_{=:~ K_i}.
\end{align*}
Note that $C_i \succeq I$ and $K_i \succeq 0$ because $M_i \succeq 0$ and $\theta > 0$.
Furthermore, the matrices $C_i$ and $K_i$ commute 
since any eigenbasis for $M_i$ is also an eigenbasis of $C_i$ and $K_i$.

To finish the proof we must show that, for distinct $a, b \in [m]$,
\tcversion{\begin{multline}}{\begin{equation*}}
    z \:\mapsto\:
    \trace \exp\Bigl(
      \log\big( C_a + z K_a \big)
    + \log\big( C_b - z K_b \big)
    \tcskip
    + \underbrace{ \sum_{i \not\in \set{a,b}}
        \log \expectover{X \sim \cD(x)}{ e^{\theta X_i M_i} } }_{=:~ L} \Bigr)
\tcversion{\end{multline}}{\end{equation*}}
is concave in a neighborhood of $0$. This follows from \Theorem{liebvariant}.
\end{proofof}

\section{Proofs of Applications}

\subsection{Rounding decompositions of the identity}
\AppendixName{Psubmodular}

Here, we give a generalization of \Theorem{isotropic} to a decomposition of the identity into matrices of arbitrary rank.
\begin{theorem}
\TheoremName{largerank}
Let $X_1,\ldots,X_m \in \Psd$ satisfy $\sum_{i=1}^m X_i = I$. 
Then there exists a subset $S \subseteq [m]$ with $\card{S} \leq n$
such that $ \sum_{i \in S} X_i/\trace X_i $
has full rank and maximum eigenvalue at most $\alpha = O(\log n / \log \log n)$.
\end{theorem}
\begin{proof}
Let $V_i$ be a matrix such that $X_i = V_i V_i \transpose$.
Define the function $r : 2^{[m]} \rightarrow \bZnneg$ by
\[
r(J)
 ~=~ \rank \Big( \sum_{j \in J} V_j V_j \transpose \Big)
 ~=~ \rank V_J,
\]
where $V_J$ is the matrix obtained by concatenating in any order all columns
from the matrices $\setst{ V_j }{ j \in J }$.
It is well-known that such a function $r$ is:
\begin{itemize}
\item \textit{Normalized:} $r(\emptyset) = 0$,
\item \textit{Monotone:} $r(I) \leq r(J)$ whenever $I \subseteq J$, and
\item \textit{Submodular:} $r(I) + r(J) \geq r(I \union J) + r(I \intersect J)$
for all $I, J \subseteq [m]$.
\end{itemize}
For any normalized, monotone, submodular function $f : 2^{[m]} \rightarrow \bR$,
its \textit{base polytope} is defined to be
\tcversion{
\begin{multline*}
    B(f) ~:=~ \bigl\{\: x \in \bRnneg^m \::\: x(J) \leq f(J) ~\:\forall J \subseteq [m],\\
    ~~\text{and}~~ x([m]) = f([m]) \:\bigr\}.
\end{multline*}
}{\[  B(f) ~:=~ \bigl\{\: x \in \bRnneg^m \::\: x(J) \leq f(J) ~\:\forall J \subseteq [m], ~~\text{and}~~ x([m]) = f([m]) \:\bigr\}. \]}

Define the vector $p \in \bR^m$ by $p_i = \trace X_i$.
Note that $p \geq 0$ and $\sum_i p_i = \trace( \sum_i X_i ) = n$,
so we can think of $p$ as defining a ``fractional multiset'' of $n$ matrices.
Intuitively, we want to ``round'' the coordinates of $p$ to integers.
To that end, define the polytope
\[
    P' ~:=~ B(r) \:\intersect\: \setst{ x }{ \floor{p} \leq x \leq \ceil{p} },
\]
where $\floor{p}$ and $\ceil{p}$ respectively denote the component-wise
floor and ceiling of the vector $p \in \bR^m$.
The polytope $P'$ is not necessarily a matroid polytope;
for example, a vector in $P'$ could have a coordinate strictly greater than $1$.

\begin{claim}
\ClaimName{pinPprime}
$p \in P'$.
\end{claim}

\begin{claim}
\ClaimName{Psubmodular}
$P := \setst{ x-\floor{p} }{ x \in P' }$ is a matroid base polytope.
\end{claim}

\Claim{pinPprime} is proven below.
\Claim{Psubmodular} is a folklore result that can be derived using reductions and contractions
of submodular functions \cite[\S 3.1(b)]{Fujishige}; see also Fujishige's
remarks on crossing submodular functions \cite[Eq.~(3.97)]{Fujishige}.

Define $A_i = X_i / \trace X_i$, $B = I$ and
\[ 
Q ~:=~ P \:\intersect\: \Bigl\{ \, x \in \bR^m \,:\, \sum_i x_i A_i \preceq B \,\Bigr\}.
\]
Setting $x = p - \floor{p}$, we have $x \in P$ by \Claim{pinPprime} and
\[
\sum_i x_i A_i
~\preceq~ \sum_i p_i \frac{X_i}{\trace X_i}
~=~ \sum_i X_i
~=~ B,
\]
so $x \in Q$.

Since $\trace A_i = 1$, we have $A_i \preceq B$.
Applying \Theorem{main}, we obtain a vector $\hat{x} \in \set{0,1}^n$ that is an extreme point of $P$,
and for which $\sum_i \hat{x}_i A_i \preceq \alpha B$.
Let $S$ be the support of $\hat{x}$.
Note that $\hat{x} + \floor{p} \in P'$.
So
\[
\card{S}
~=~ \sum_{i=1}^m \hat{x}_i
~\leq~ \sum_{i=1}^m (\hat{x}_i + \floor{p_i})
~\leq~ r([m])
~=~ n
\]
and $\sum_{i \in S} X_i / \trace X_i \preceq \alpha B$ as required.
\end{proof}

\begin{proofof}{\Claim{pinPprime}}
The box constraint $\floor{p} \leq p \leq \ceil{p}$ is trivially satisfied.
We have noted above that $\sum_i p_i = n$,
so the constraint $p([m]) \leq r([m]) = n$ is also satisfied.

It remains to show that $\sum_{i \in I} p_i \leq r(I)$ for all $I$.
For any positive semidefinite matrix,
the average of the non-zero eigenvalues is a lower bound on the maximum eigenvalue, so
\[
\frac{ \trace( \sum_{i \in I} X_i ) }
     { \rank( \sum_{i \in I} X_i ) }
     ~\leq~ \Bigl\| \sum_{i \in I} X_i \Bigr\|
     ~\leq~ \Bigl\| \sum_{i=1}^m X_i \Bigr\|
     ~=~ 1.
 \]
Thus
$
\sum_{i \in I} p_i 
 = \trace( \smallsum{i \in I}{} X_i )
 \leq \rank( \smallsum{i \in I}{} X_i )
 = r(I).
$
This proves that $p \in P$.
\end{proofof}

\subsection{Thin trees}
\AppendixName{transitive}

\begin{proofof}{\Theorem{conductanceweightedtree}}
Recall the notation defined in \Section{prelim}.
For $e = uv \in E$, define vectors $x_e = L_G^{+/2} (e_u - e_v)$
and $w_e = x_e / \norm{x_e}$. 
Then $R_e = \norm{x_e}^2$; let $p_e = R_e / (n-1)$.
It is well-known \cite{Bollobas} that the vector of effective resistances describes the edge marginals of the uniform spanning tree, and hence that $\sum_e p_e = 1$.
Then, following the argument of Spielman and Srivastava \cite{SS08},
\begin{align*}
\sum_{e \in E} p_e w_e w_e \transpose
&~=~
\frac{1}{n-1} \sum_{e \in E} x_e x_e \transpose \\
&~=~
\frac{1}{n-1} L_G^{+/2} \Big( \sum_{e \in E} (e_u \!-\! e_v) (e_u \!-\! e_v) \transpose \Big) L_G^{+/2} \\
&~=~
\frac{1}{n-1} I_{\image L_G}.
\end{align*}
We view the vectors $\setst{ w_e }{ e \in E }$ as $(n-1)$-dimensional vectors
in their linear span and apply \Theorem{isotropic}.
This gives a set $T \subseteq E$ of size $n-1$ such that $\setst{ w_e }{ e \in T }$ 
is linearly independent and
\[
    \sum_{e \in T} w_e w_e \transpose
    ~\preceq~   O(\log n / \log \log n) \cdot I_{\image L_G}.
\]
The first two conditions imply that the edges in $T$ form a spanning tree
on the vertex set $V$.
Then since $R_e = \norm{x_e}^2$, we have
\[
    \sum_{e \in T} \frac{x_e x_e \transpose}{ R_e }
    ~\preceq~   O(\log n / \log \log n) \cdot I_{\image L_G}.
\]
Equivalently,
\[
    \sum_{uv \in T} \frac{ (e_u \!-\! e_v) (e_u \!-\! e_v) \transpose }{ R_{uv} }
    ~\preceq~ O(\log n / \log \log n) \cdot L_G.
\]
Since we assume that $\conduc \leq C_e = 1/R_e$ for every edge $e$, we obtain
\[
    L_T ~=~ \sum_{uv \in T} (e_u \!-\! e_v) (e_u \!-\! e_v) \transpose
    ~\preceq~ O\Big( \frac{ \log n }{ \conduc \log \log n } \Big) \cdot L_G.
\]
So $T$ is $O\big( \frac{\log n}{\conduc \log \log n} \big)$-spectrally-thin.
\end{proofof}

\begin{proofof}{\Corollary{transitive}}
By the nearly equal resistances assumption, $R_e = O(\frac{n-1}{\card{E}})$ for every edge $e$.
On the other hand, the connectivity $\connec$ is at most the average degree, which is $2 \card{E}/n$.
Thus $R_e = O(1/\connec)$ for every edge $e$.
The result now follows from \Theorem{conductanceweightedtree}.
\end{proofof}

\subsubsection{Proof of {\protect \Theorem{nospectrallythintree}}}
\AppendixName{nospectrallythintree}

Assume $n$ is a multiple of $4$.
We define a graph that is related to an example of Boyd and Pulleyblank \cite[p.~180]{BP}.
There are two disjoint cycles, each of length $n/2$.
Let us number the vertices in the first cycle as $1,\ldots,n/2$
and the vertices in the second cycle as $n/2+1,\ldots,n$.
Add a matching where the $i\th$ edge connects the $i\th$ vertex 
in the first cycle and the $i\th$ vertex in the second cycle.
The edges in the cycles each have weight $w_c := k/2$
and the edges in the matching each have weight $w_m := 2k/n$.
Obviously this weighted graph has connectivity at least $k$.

Let $T$ be any subtree of $G$, without any weights on the edges of $T$.

\begin{claim}
Suppose that $T$ uses only a single matching edge.
There exists a vector $z$ such that
$$
\frac{z \transpose L_T z}{z \transpose L_G z} ~=~ \Omega\Big(\frac{\sqrt{n}}{k}\Big).
$$
\end{claim}
\begin{proof}
Without loss of generality, $\set{n/4,3n/4}$ be the matching edge used by $T$.
Let $\alpha=n^{-0.5}$ and $c=1-\alpha$.
Define the vector $z$ where
$$
z_i ~=~ 
\begin{cases}
c^{\abs{n/4-i}} &\quad (i \leq n/2) \\
0 &\quad (i>n/2).
\end{cases}
$$

\noindent
\textit{Numerator:}
The numerator is
$z \transpose L_T z = \sum_{uv \in E} (z_u-z_v)^2 \geq (z_{n/4}-z_{3n/4})^2 = 1$.

\vspace{3pt}
\noindent
\textit{Denominator:}
To evaluate $z \transpose L_G z$, we separately consider the cycle edges and matching edges.
The contribution from the matching edges is
\begin{align*}
C_m~:=~ w_m \cdot \sum_{i=1}^{n/2} (z_i - z_{n/2+i})^2 
    ~<~ 2 w_m \cdot \sum_{i \geq 0} c^{2i} 
    ~<~ \frac{2 w_m}{1-c} 
    ~=~ \frac{2 w_m}{\alpha}.
\end{align*}
The contribution from the cycle edges is
\begin{align*}
C_c~:=~& w_c \sum_{i=2}^{n/2} (z_{i-1} - z_i)^2 ~+~ w_c (z_1-z_{n/2})^2 
    ~<~  2 w_c \sum_{i \geq 1} (c^{i-1} - c^i)^2  \\
    ~=~& 2 w_c (1-c)^2 \sum_{i \geq 0} c^{2i} 
    ~=~  2 w_c \frac{(1-c)^2}{1-c^2} 
    ~<~  2 w_c \frac{(1-c)^2}{1-c} 
    ~=~  2 w_c \alpha.
\end{align*}
Since $\alpha = n^{-0.5}$, we get $C_m = O(k/\sqrt{n})$ and $C_c = O(k/\sqrt{n})$,
so $z \transpose L_G z = O(k/\sqrt{n})$.
\end{proof}

\begin{claim}
Suppose that $T$ uses $m > 1$ matching edges.
There exists a vector $z$ such that
$$
\frac{z \transpose L_T z}{z \transpose L_G z} ~=~ \Omega\Big(\frac{\sqrt{n}}{k}\Big).
$$
\end{claim}
\begin{proof}
Let the matching edges used by $T$ be $\set{a_1,b_1}, \set{a_2,b_2}, \ldots, \set{a_m,b_m}$.
Define the vector $z$ by
$$
z_i ~=~ 
\begin{cases}
c^{\min_j d_1(i,j)} &\quad (i \leq n/2) \\
0 &\quad (i>n/2)
\end{cases}
$$
where $d_1$ denotes distance in the first cycle.

\vspace{3pt}\noindent\textit{Numerator:}
As before, every matching edge used by $T$ contributes at least $1$,
so $z \transpose L_T z \geq m$.

\vspace{3pt}\noindent\textit{Denominator:}
Obviously $z \transpose L_G z$ is no more than $m$ times what it would be if $T$ used
only a single matching edge.
That is, $z \transpose L_G z \leq O(m k/\sqrt{n})$.
\end{proof}

\subsection{Column-subset selection}
\AppendixName{CSS}

\begin{proofof}{\Claim{CSpinP}}
The proof is analogous to \Claim{pinP}.
As before, let $r : 2^{[m]} \rightarrow \bZnneg$ be defined by
\[
    r(S)
    ~:=~ \operatorname{dim} \operatorname{span} \setst{ a_i }{ i \in S }
    ~=~ \rank A_S A_S \transpose.
\]
Then 
\tcversion{
\begin{align*}
P ~:=~ \Bigl\{\: x \in \bRnneg^m ~:~ &x(J) \leq r(J) ~~\forall J \subseteq [m],\\
&\text{and}~~ x([m]) = \floor{\stable(A)} \:\Bigr\}.
\end{align*}
}{
\[ P ~:=~ \Bigl\{\: x \in \bRnneg^m ~:~ x(J) \leq r(J) ~~\forall J \subseteq [m], ~\text{and}~ x([m]) = \floor{\stable(A)} \:\Bigr\}. \]
}

For any set $J \subseteq [m]$, we have
\begin{align*}
    r(J) ~&=~ \rank A_J A_J \transpose
         ~\geq~ \frac{ \trace A_J A_J \transpose }{ \norm{A_J A_J \transpose} }
         ~\geq~ \frac{\card{J}}{ \norm{A}^2 }\\
         ~&\geq~ \floor{\frac{m}{ \norm{A}^2 }} \frac{\card{J}}{m}
         ~=~ p(J).
\end{align*}
Since $\sum_i p_i = \floor{\stable(A)}$, we have $p \in P$.
\end{proofof}

\section{Proof of Theorem~\ref{LIEBVARIANT}}
\AppendixName{epstein}

\newcommand{\Cp}{\mathbb{C}_{++}}
\newcommand{\Cnneg}{\mathbb{C}_{+}}
\newcommand{\impos}{\mathcal{I}_{++}}
\newcommand{\imnneg}{\mathcal{I}_{+}} 
\newcommand{\imneg}{\mathcal{I}_{--}}
\newcommand{\herm}[1]{\mathbb{H}^{#1}}
\newcommand{\complexmat}[1]{M_{#1}(\mathbb{C})}
\newcommand{\lambdamin}{\lambda_{\textit{min}}}

The outline of this proof follows a proof of Lieb's theorem presented by Epstein~\cite{Epstein}.
Epstein's proof proceeds via complex analytic techniques, and in particular makes use of some powerful results involving \emph{Herglotz functions}
(see, e.g., \cite{Bhatia, GT}).
While an effort has been made to make the treatment here accessible, a modicum of complex analysis will be assumed; a standard reference is~\cite{Rudin}.

For a complex number $z$, let $\Re z$ and $\Im z$ respectively denote the real
and imaginary parts of $z$.
Let $\Cp = \{ z \in \bC \mid \Im z > 0 \}$ denote the open upper half-plane, and $\Cnneg$ the closed upper half-plane. Define $\bC_{--}$ and $\bC_-$ in the obvious corresponding way.

\begin{definition}
A function $g : \Cp \rightarrow \bC$ is called a \emph{Herglotz function} 
(or \emph{Pick function}) if it is analytic on $\Cp$ and $g(\Cp) \subseteq \Cp$.
\end{definition}

For example the map $z \mapsto az+b$ is Herglotz if $a \in \bRnneg$ and $b \in \Cnneg$.
The maps $z \mapsto -1/z$ and $z \mapsto \log z$ are also Herglotz.

A key reason that Herglotz functions will be useful is the following classical theorem (see, e.g.,
\cite[Eq.~V.42]{Bhatia} or \cite[p.~542]{HornJohnsonTopics}).

\begin{theorem}[Herglotz-Nevanlinna-Riesz representation theorem]
    \label{thm:hnr}
    For any Herglotz function $g$, there exists $a \in \bR$, $b \in \bRnneg$
    and a positive Borel measure $\mu$ on $\bR$, with $\int_{\mathbb{R}} \frac{1}{t^2+1}d\mu(t) < \infty$, s.t.
\begin{equation}\label{eq:hnr} g(z) ~=~ a + bz + \int_{\mathbb{R}} \Big( \frac{1}{t - z} - \frac{t}{1+t^2}\Big) \, d\mu(t) \quad \forall z \in \Cp. 
\end{equation}
\end{theorem}
Roughly speaking, this provides a description of a Herglotz function through its boundary (the real
line); since the function may diverge as it approaches the real line, the generality of a measure (which may have atoms) is needed.

The relevance of this theorem to our purposes comes from the following:
\begin{lemma}[Implicit in~\cite{Epstein}]\label{lem:concave}
    Let $D$ be a domain\footnote{Recall that a domain is an open, connected set.} in $\bC$ containing $\mathbb{C}_{--} \cup \{0\}$. 
    Suppose $f: D \to \bC$ is analytic, its restriction to $D \cap \bR$ is real-valued, and moreover the function $g$ on $\Cp$ defined by $g(z) = zf(1/z)$ is a Herglotz function.
 Then the restriction of $f$ to $D \cap \bR$ is concave in some neighborhood of the origin.
\end{lemma}
\begin{proof}
    Since $0 \in D$ and $D$ is open, there exists some $\tau > 0$ so that the interval $[-\tau, \tau] \subset D$.
    Let $D'$ be the image of $D$ under the map $z \mapsto 1/z$; so $D'$ contains $\Cp \cup [\tau^{-1}, \infty) \cup (-\infty, -\tau^{-1}]$. We may think of $g$ as being defined on all of $D'$.
    Let $\mu$ be the positive Borel measure associated with $g$ by Theorem~\ref{thm:hnr}.
    This measure can be thought of as the limit of $\Im g(z)$ as $z$ approaches the real line, in the appropriate distributional sense: this is known as the Stieltjes inversion formula; see, e.g.,~\cite[Thm.~V.4.12]{Bhatia}, \cite[Thm.~2.2]{GT}.
We will use only the following consequence: 
\begin{quote}
    If for some open interval $I \subseteq \mathbb{R}$, 
    \[ \lim_{\epsilon \downarrow 0} \Im g(w + i\epsilon) = 0 \quad \text{for all } w \in I, \] then $\mu(I) = 0$.
\end{quote}
 We deduce that $\mu$ is supported on $[-\tau^{-1}, \tau^{-1}]$, since $\lim_{\epsilon \downarrow 0} \Im g(w + i\epsilon) = \Im g(w) = 0$ for all $w \in D' \cap \bR$.

 Expressing $f$ in terms of the Herglotz-Nevanlinna-Riesz representation of $g$, we have that
\[ f(z) ~=~ az + b + \int_{-\tau^{-1}}^{\tau^{-1}} \frac{z^2}{zt-1} \,d\mu(t). \]
(Note that the final term of \eqref{eq:hnr} can be folded into the constant $a$ --- since $\mu$ is Borel and has bounded support, it is finite.)
Now calculate the second derivative of $f$, considered as a real-valued function on $D \cap \bR$:
\begin{align*}
 f''(x) &~=~ \int_{-\tau^{-1}}^{\tau^{-1}} \frac{d^2}{dx^2}\left(\frac{x^2}{xt-1}\right) \,d\mu(t)
 \iftc{\\&}
 ~=~ \int_{-\tau^{-1}}^{\tau^{-1}} \frac{2}{(xt-1)^3} \,d\mu(t).
\end{align*}
So for all $x \in (-\tau, \tau)$, $f''(x) < 0$, and so $f$ (as a real-valued function on $D \cap \bR$) is concave in the neighborhood of $0$.
\end{proof}

We will apply \Lemma{concave} with $f$ as in the statement of \Theorem{liebvariant}:
\[ f(z) = \trace \exp\Big( L + \log( C_1 + z K_1 ) + \log( C_2 - z K_2 ) \Big).\]
In order to extend our definition of $\log$ beyond symmetric matrices, we use (again following~\cite{Epstein}) the Cauchy integral description
\[ \log C = \int_{0}^\infty \frac{1}{t+1} - (t + C)^{-1}\,dt;\]
this is well-defined as long as $C$ has no nonpositive eigenvalues. 
As our domain $D$, we take $\bC_{--} \cup B_\epsilon$, where $B_\epsilon$ is an open ball around the origin of radius $\epsilon := \tfrac12\min\{\lambdamin(C_1)\|K_1\|^{-1}, \lambdamin(C_2)\|K_2\|^{-1}\}$. 
This ensures that 
\begin{lemma}\label{lem:welldefined}
  The function $f$ is well-defined and analytic on $D$.
\end{lemma}
For convenience, we withhold the proof until the end of this section.

To deduce that $f$ is concave by Lemma~\ref{lem:concave}, we must show that $g$ defined by $g(z) = zf(1/z)$ is Herglotz.
We have
\begin{align*}
    &g(z) = z \cdot f(1/z)\\
         &~= z \trace \exp \Big( L + \log(C_1 + K_1/z) + \log(C_2 - K_2/z) \Big)\\
         &~= \trace \exp \Big( \log(zI) + L + \log(C_1 \!+\! K_1/z) + \log(C_2 \!-\! K_2/z) \Big)\\
         &~= \trace \exp \Big( L + \log(C_1z + K_1) + \log(C_2 + ( -1/z)K_2)\Big).
\end{align*}

We will work with complex matrices for the remainder of this section, so let $\complexmat{n}$ denote the space of $n\times n$ complex matrices, and $\herm{n}$ the space of $n \times n$ Hermitian matrices.
We will make use of operator formalism on occasion; in particular, the identity $I$ will generally be omitted, and so for a scalar $w \in \mathbb{C}$, $wI$ will be written as simply $w$.

An arbitrary matrix $C \in \complexmat{n}$ has a unique decomposition $C = P + iQ$, where $P, Q \in \herm{n}$.
This is obtained by taking $P = \tfrac12(C + C^*)$ and $Q = \tfrac{1}{2i}(C - C^*)$, where $C^*$ denotes the adjoint (conjugate transpose) of $C$.
By analogy with the scalar ($n\!=\!1$) case, we say that $P$ is the ``real part'' of $C$, denoted by
$\Re C$, and that $Q$ is the ``imaginary part'' of $C$, denoted by $\Im C$. (Note that this has nothing to do with the entry-wise real and imaginary parts of the matrix.)

This analogy to the scalar case provides a lot of helpful intuition, and so at this point we will sketch a version of the proof for $n=1$. 
The full argument will follow the same essential steps, though the generalization is not completely straightforward.
The scalar analog of a Hermitian matrix is a real number,
and the scalar analog of a positive definite matrix is a positive number;
so we consider the function $h : \bC \rightarrow \bC$ defined by
\[ h(z) ~=~ \exp\big(l + \log(c_1 z + k_1) + \log(c_2 + (-1/z)k_2)\big), \]
with real parameters $l \in \bR$, $k_1, k_2 \geq 0$ and $c_1, c_2 > 0$. 
Then 
\[ \Im \log(c_1z + k_1) ~=~ \arg(c_1z + k_1) ~\in~ (0, \arg z]. \]
Similarly, 
\[ \Im \log(c_2 + (-1/z)k_2) \:=\: \arg(c_2 + (-1/z)k_2) \:\in\: [0, \arg (-1/z)). \]
Since $\arg (-1/z) = \pi - \arg z$, we obtain
\[ \Im(l +  \log(c_1z + k_1) + \log(c_2 + (-1/z)k_2)) ~\in~ (0, \pi). \]
Since $\Im e^{a+ib} = e^a \sin b$ for $a,b \in \bR$, we deduce that $\Im h(z) > 0$, as required.

\medskip
We now resume the argument for the case $n>1$. Define
\begin{align*}
\impos  &~=~ \setst{ C \in \complexmat{n} }{ \Im C \succ 0 } \\
\imnneg &~=~ \setst{ C \in \complexmat{n} }{ \Im C \succeq 0 }.
\end{align*}
Much of the argument revolves around noting that $\impos$ is closed under various operations.
For example, if $C, A \in \impos$ then clearly $A+C \in \impos$.
The following is less straightforward:
\begin{lemma}[{\cite[{\upshape pp.~318--319}]{Epstein}}]\label{lem:closure}
    For any $C \in \impos$,
    \begin{enumerate}[(i)]
        \item\label{item:inv} $-C^{-1} \in \impos$, and
        \item\label{item:log} $0 \prec \Im \log C \prec \pi$.
    \end{enumerate}
\end{lemma}
We refer to~\cite{Epstein} for the proofs, but we again note the intuition by analogy with the $n=1$ case, where $C$ is just an element of $\Cp$. 
Then $C=re^{i\theta}$ for some $r > 0$ and $0 < \theta < \pi$; so $-C^{-1} = r^{-1}e^{i(\pi - \theta)} \in \Cp$ and $\log C = \log r + i\theta$.

A crucial lemma will be the following: 
\begin{lemma}\label{lem:stronglog}
     Let $A, B \in \herm{n}$ satisfy $A, B \succeq 0$, where in addition at least one of $A$ and $B$ are strictly positive definite.
     Then for any $z \in \Cp$, $\log (A+Bz)$ is defined and
\[ 0 ~\preceq~ \Im \log (A+Bz) ~\preceq~ \arg z. \]
Moreover, if $A \succ 0$, then the left inequality is strict, and if $B \succ 0$, the right inequality is strict.

\end{lemma}
\begin{proof}
    We first observe that the conditions imply that $A+Bz$ has no nonpositive real eigenvalues, and hence that the logarithm is well defined.
    It suffices to show that $A+Bz$ is nonsingular, since we can apply the same argument to $A' + Bz$, where $A' = A + t$ for any $t \geq 0$.

    If $B \succ 0$, then $B^{1/2}$ exists and is positive definite. 
    Thus
    \[ A + Bz = B^{1/2}(\underbrace{B^{-1/2}AB^{-1/2}}_{=:\, Q} + z)B^{1/2}. \]
    But $Q$ is Hermitian (as can be seen since $B^{-1/2}$ and $A$ are Hermitian) and so it has real spectrum; thus since $\Im z > 0$, $0$ is not in the spectrum of $Q+z$.
    Hence $Q + z$ and so also $A+Bz$ are invertible.

    If instead $A \succ 0$, then
    \[ A + Bz = zA^{1/2}(1/z + A^{-1/2}BA^{-1/2})A^{1/2}, \]
    and similar reasoning applies.

    Suppose first that $B \succ 0$.
    Then $A+Bz \in \impos$, and so by \Lemma{closure}~(\ref{item:log}) we immediately have that $\Im\log(A+Bz) \succ 0$.
    Now if $B \succeq 0$ but is not positive definite, then $B+\epsilon \succ 0$ for any $\epsilon > 0$, and so $\Im\log(A+(B+\epsilon)z) \succ 0$.
    Since $\log (A+Bz)$ is well defined, we have by continuity that 
    \[ \Im \log(A+Bz) = \lim_{\epsilon \downarrow 0} \Im \log(A+(B+\epsilon)z) \succeq 0. \]
    This completes the proof of the left inequality.

    For the right inequality, suppose first that $A \succ 0$.
Since $\arg z = \Im \log z$, our goal is to show that
\[ \Im (\log z - \log(A + Bz) ) \succ 0, \]
or equivalently (using that $A+Bz$ is nonsingular)
\[ \Im \log ((A/z+B)^{-1}) \succ 0. \]
Now since $-1/z \in \Cp$, it follows that $\Im(-A/z) \succ 0$.
Since $\Im B = 0$, we obtain that $-A/z - B \in \impos$.
Thus $(A/z+B)^{-1} \in \impos$ by Lemma~\ref{lem:closure}~(\ref{item:inv}), and so the result follows by Lemma~\ref{lem:closure}~(\ref{item:log}).
If $A \succeq 0$ but $A$ is not positive definite, we apply a limiting argument as before to deduce that $\Im \log(A+Bz) \succeq 0$.
\end{proof}

We will also need the following result:
\begin{lemma}[\cite{Epstein}]\label{lem:exp}
    If\/ $0 \prec \Im C \prec \pi$, then $\trace \exp C \in \Cp$.
\end{lemma}
We omit the proof, which proceeds by first showing that the spectrum of $C$ is contained in the strip
$\setst{ z \in \mathbb{C} }{ 0 < \Im z < \pi }$,
and then using the spectral mapping theorem to deduce that the spectrum of $\exp C$ lies in $\Cp$.

\begin{lemma}
    The function $g$ is Herglotz.
\end{lemma}
\begin{proof}
    Take any $z \in \Cp$.
    By \Lemma{stronglog}, we have that
\begin{align*}
0 &\preceq \Im \log (C_1z + K_1) \prec \arg z\\
\text{and} \qquad 0 &\prec \Im \log (C_2 + (-1/z)K_2) \preceq \arg(-1/z).
\end{align*}
    Since $\arg(-1/z) = \pi - \arg z$, we obtain that
    \[ 0 \prec \Im \bigl(L + \log(C_1z + K_1) + \log(C_2 + (-1/z)K_2)\bigr) \prec \pi. \]
    Thus by \Lemma{exp}, $g(z) \in \Cp$.
    Hence $g$ is indeed Herglotz.
\end{proof}

Applying \Lemma{concave}, and observing the proof of Lemma~\ref{lem:welldefined} below, \Theorem{liebvariant} has been proved.

\begin{proofof}{Lemma~\ref{lem:welldefined}}
    Firstly, if $z \notin \bR$, then either $C_1 + zK_1 \in \impos$, or $-(C_1 + zK_1) \in \impos$. 
   Thus, as observed by Epstein, $\log(C_1 + zK_1)$ is defined; indeed, we already proved more in Lemma~\ref{lem:stronglog}.
   The same is true for $\log(C_2 - zK_2)$.

    Now suppose $z \in (-\epsilon, \epsilon)$.
Then 
\[ C_1 + zK_1 \succeq C_1 - \epsilon\|K_1\| \succeq C_1 - \tfrac12\lambdamin(C_1) \succ 0. \]
Similarly $C_2 - zK_2 \succ 0$. 
\end{proofof}

\section{Weaker Proof of Theorem~\ref{LIEBVARIANT}}
\AppendixName{lieb}

In this appendix we prove \Theorem{liebvariant}, under the additional
hypothesis that $C_i$ \& $K_i$ commute.
This suffices to prove \Lemma{pipage}. 
The argument builds on Lieb's original proof~\cite{Lieb} of \Theorem{lieb}.

\begin{theorem}
\TheoremName{specialcase}
Let $L \in \Sym$, $C_1, C_2 \in \Pd$ and $K_1, K_2 \in \Psd$ be such that
$C_1$ \& $K_1$ commute, and that $C_2$ \& $K_2$ commute.
Then
\begin{equation}
\label{eq:g}
f(z) \::=\: \trace \exp\Big( L + \log( C_1 + z K_1 ) + \log( C_2 - z K_2 ) \Big)~
\end{equation}
is concave in a neighborhood of $0$.
\end{theorem}

First we need some preliminary definitions.
For $x,y \geq 0$,
define the logarithmic mean and binomial mean as follows:
\begin{align*}
\LM(x,y) &~=~
\begin{cases} \frac{x-y}{\log x - \log y} &\quad(x \neq y) \\ x &\quad(\text{otherwise}) \end{cases} \\
\AGM(x,y) &~=~ \Big( \frac{x+y}{2} + \sqrt{xy}\Big)/2
          ~=~ \Big( \frac{\sqrt{x}+\sqrt{y}}{2} \Big)^2.
\end{align*}

\begin{theorem}[Carlson~\cite{Carlson}, Bhatia~\cite{BhatiaMean}]
\TheoremName{Carlson}
For $x,y \geq 0$, 
\[ \sqrt{xy} ~\leq~ \LM(x,y) ~\leq~ \AGM(x,y) ~\leq~ (x+y)/2. \]
\end{theorem}

For any $X \in \Pd$, define the operators $T_X, R_X : \Sym \rightarrow \Sym$ by
\begin{align*}
T_X(Y) &~:=~ \int_{0}^{\infty} (X+tI)^{-1} Y (X+tI)^{-1} \, dt \\
R_X(Y) &~:=~ 2 \int_{0}^{\infty} (X+tI)^{-1} Y (X+tI)^{-1} Y (X+tI)^{-1} \, dt.
\end{align*}

\newcommand{\FACTCOMM}{\textup{\textbf{(P1)}}}
\newcommand{\FACTINV}{\textup{\textbf{(P2)}}}
\newcommand{\FACTDIAG}{\textup{\textbf{(P3)}}}
\newcommand{\FACTPOS}{\textup{\textbf{(P4)}}}

\begin{claim}
Let $X \in \Pd$ and $Y \in \Sym$.
\begin{itemize}
\item \FACTCOMM\textup{:}
If $X$ and $Y$ commute then \tcskip
$T_X(Y)=YX^{-1}$ and $R_X(Y)=Y^2 X^{-2}$.

\item \FACTINV\textup{:}
The inverse of $T_X$ is the operator $T_X^{-1}$ where \tcskip
$T_X^{-1}(Y) = \int_{0}^1 X^t Y X^{1-t} \, dt $.

\item \FACTDIAG\textup{:}
In a basis in which $X$ is diagonal, we have \tcskip
$\big(T_X^{-1}(Y)\big)_{i,j} = Y_{i,j} \cdot \LM(X_{i,i},X_{j,j})$.

\item \FACTPOS\textup{:}
$T_X$ is a positive map, i.e., $T_X(Y) \in \Psd$ whenever $Y \in \Psd$.
\end{itemize}
\end{claim}
\begin{proof}
See Lieb \cite{Lieb} p.~277, and Ohya and Petz \cite{OP} Eq.~(3.7) and p.~49.
\end{proof}

\begin{claim}
For any $C \in \Pd$, $K \in \Sym$ and $x \in \bR$,
\[
    \log(C+xK) ~=~ \log C + x T_C(K) - \frac{1}{2} x^2 R_C(K) + O(x^3).
\]
\end{claim}
\begin{proof}
See Lieb \cite{Lieb} equations (3.6) and (3.9), and Ohya and Petz \cite[p.~53]{OP}.
\end{proof}

\begin{claim}
\ClaimName{exp}
Let $L \in \Sym$, $C_1, C_2 \in \Pd$ and $K_1, K_2 \in \Sym$.
Define $M = \exp(L + \log C_1 + \log C_2)$.
Then
\begin{align*}
&\exp\Big( L + \log( C_1 + z K_1 ) + \log( C_2 - z K_2 ) \Big) \\
&\quad=~ M + z \int_0^1 M^{1-s} \Big( T_{C_1}(K_1) - T_{C_2}(K_2) \Big) M^{s} \, ds \\
&\quad+~ z^2 \Bigg( -\frac{1}{2} \int_0^1 M^{1-s} \Big( R_{C_1}(K_1) + R_{C_2}(K_2) \Big) M^s \, ds \\
&\quad+~ \int_0^1 \int_0^s M^{1-s}
                \Bigl( T_{C_1}(K_1) - T_{C_2}(K_2) \Bigr) M^{s-u}\tcskip
                \Bigl( T_{C_1}(K_1) - T_{C_2}(K_2) \Bigr) M^u ~ du \, ds \Biggr)
 ~+~ O(z^3).
\end{align*}
\end{claim}
\begin{proof}
Similar to Ohya and Petz \cite[p.~53]{OP}.
\end{proof}

\begin{proofof}{\Theorem{specialcase}}
The theorem is equivalent to $0 \leq \frac{d^2 f}{dz^2} \big|_{z=0}$
(assuming that this derivative exists).
From \Claim{exp} we have
\begin{align}
\EquationName{deriv}
&{\displaystyle \frac{d^2 f}{dz^2} \bigg|_{z=0} }
\\\nonumber
&=~ {-} \trace M \big( R_{C_1}(K_1) + R_{C_2}(K_2) \big) ~+
\iftc{\\\nonumber&}
\trace {\textstyle \int_0^1} \big( T_{C_1}(K_1) \!-\! T_{C_2}(K_2) \big) M^y \big( T_{C_1}(K_1) \!-\! T_{C_2}(K_2) \big) M^{1-y} \, dy
\\\nonumber
&=~ {-} \trace M \big( R_{C_1}(K_1) + R_{C_2}(K_2) \big) ~+
\iftc{\\\nonumber&}
\trace \big( T_{C_1}(K_1) \!-\! T_{C_2}(K_2) \big) T_M^{-1} \big( T_{C_1}(K_1) \!-\! T_{C_2}(K_2) \big).
\end{align}
From \FACTCOMM\ and the assumption that $C_i$ and $K_i$ commute
we have $R_{C_i}(K_i) = T_{C_i}(K_i)^2$. 
So the assertion of the theorem is equivalent to
\begin{multline}
\EquationName{equivassertion}
\trace M \, T_{C_1}(K_1)^2 + \trace M \, T_{C_2}(K_2)^2
~\geq~ \iftc{\\}
\trace \Big( T_{C_1}(K_1) - T_{C_2}(K_2) \Big) T_M^{-1}
         \Big( T_{C_1}(K_1) - T_{C_2}(K_2) \Big).
\end{multline}
We will prove the more general statement that for all $M \in \Pd$ and $X, Y \in \Psd$, 
\begin{equation}
\EquationName{diagonalcase}
\trace M X^2 ~+~ \trace M Y^2 ~\geq~ \trace (X-Y) T_M^{-1}(X-Y).
\qquad
\end{equation}
This implies \eqref{eq:equivassertion} by our assumption that $K_1, K_2 \in \Psd$
and \FACTPOS.

The preceding discussion is basis-independent.
It is now convenient to fix a basis in which $M$ is diagonal and to 
view $M$, $X$ and $Y$ as matrices in that basis.
Let us denote the diagonal entries of $M$ by $\lambda_i = M_{i,i}$;
these are positive since we assume $M \in \Pd$.
By \FACTDIAG, the right-hand side of \eqref{eq:diagonalcase} is
\begin{align}
\trace (X - Y) T_M^{-1}( X - Y )
~&=~ \sum_{i,j} \: \LM(\lambda_i,\lambda_j) \cdot ( X_{i,j} - Y_{i,j} \big)^2 \notag \\
~&\leq~ \sum_{i,j} \: \AGM(\lambda_i,\lambda_j) \cdot ( X_{i,j} - Y_{i,j} \big)^2,
\EquationName{AGM}
\end{align}
by \Theorem{Carlson}.
We may rewrite the right-hand side as
\begin{multline}
\sum_{i,j}
    \Big( \frac{\lambda_i}{4} + \frac{\lambda_j}{4} + \frac{\sqrt{\lambda_i \lambda_j}}{2} \Big)
    ( (X_{i,j})^2 + (Y_{i,j})^2 - 2 X_{i,j} Y_{i,j} ) \\
\EquationName{lotsoftrace}
~=~ \frac{\trace M X^2}{2} + \frac{\trace M^{1/2} X M^{1/2} X}{2} \iftc{\\}
   + \frac{\trace M Y^2}{2}\, + \frac{\trace M^{1/2} Y M^{1/2} Y}{2} \iftc{\\}
   - \trace M X Y - \trace M^{1/2} X M^{1/2} Y
\end{multline}
by repeatedly using the observation
\[
\sum_{i,j} D_{i,i} P_{i,j} Q_{i,j} E_{j,j}
    ~=~ \trace DPEQ
    ~=~ \trace EPDQ
\]
for all $D, E \in \Diag$, $P, Q \in \Sym$.

Thus, combining \eqref{eq:diagonalcase}, \eqref{eq:AGM} and \eqref{eq:lotsoftrace},
it suffices to prove
\begin{multline*}
    \trace M X^2
    - \trace M^{1/2} X M^{1/2} X
    + \trace M Y^2
    - \trace M^{1/2} Y M^{1/2} Y\iftc{\\}
    ~\geq~
    - 2 \trace M X Y
    - 2 \trace M^{1/2} X M^{1/2} Y
\end{multline*}
for every $M, X, Y \in \Psd$.

Since that inequality is invariant under
choice of orthonormal basis, and since $\trace M^{1/2} X M^{1/2} Y \geq 0$,
it suffices to prove
\begin{equation}
\EquationName{rewritten}
\begin{split}
&\trace X D^2 X  - \trace X D X D
+ \trace Y D^2 Y  - \trace Y D Y D \iftc{\\&\quad }
~\geq~ {-2} \trace X D^2 Y
\quad~~\forall D \in \Diag ,~ \forall X, Y \in \Psd.
\end{split}
\end{equation}
Denote the diagonal entries of $D$ by $d_i = D_{i,i}$.
Then
\begin{align*}
  \trace X D^2 X - \trace XDXD
 &~=~ \frac{1}{2} \sum_{i,j} X_{i,j}^2 (d_i^2+d_j^2) - \sum_{i,j} X_{i,j}^2 d_i d_j
 \iftc{\\&}
 ~=~ \frac{1}{2} \sum_{i,j} X_{i,j}^2 (d_i - d_j)^2.
\end{align*}
So the left-hand side of \eqref{eq:rewritten} equals
\[
\sum_{i,j} \frac{X_{i,j}^2 + Y_{i,j}^2}{2} (d_i - d_j)^2
~\geq~ 
\sum_{i,j} |X_{i,j} Y_{i,j}| \cdot (d_i - d_j)^2,
\]
by the arithmetic-mean geometric-mean (AM-GM) inequality.
The right-hand side of \eqref{eq:rewritten} is
\begin{align*}
{-} 2 \trace(X D^2 Y) 
 &~=~ {-} \trace(X D^2 Y) - \trace(D^2 X Y)
\iftc{\\&}
 ~=~ {-} \sum_{i,j} X_{i,j} Y_{i,j} (d_i^2 + d_j^2).
\end{align*}
So, to prove \eqref{eq:rewritten}, it suffices to prove that
\begin{equation}
\EquationName{weirdabs}
\sum_{i,j} |X_{i,j} Y_{i,j}| \cdot (d_i - d_j)^2
~\geq~ 
- \sum_{i,j} X_{i,j} Y_{i,j} (d_i^2 + d_j^2).
\end{equation}

We will prove the more general inequality
\begin{equation}
\EquationName{moregeneral}
\sum_{i,j} | Z_{i,j} | \cdot (d_i - d_j)^2
~\geq~
- \sum_{i,j} Z_{i,j} (d_i^2 + d_j^2)
\quad \forall d \in \bR^n ,~ \forall Z \in \Psd.
\end{equation}
This implies \eqref{eq:weirdabs} by letting $Z = X \circ Y$
(the Hadamard product of $X$ and $Y$),
which is positive semidefinite by the Schur product theorem \cite[p.~23]{Bhatia}.
Rearranging, \eqref{eq:moregeneral} becomes
\begin{equation}
\EquationName{moregeneral2}
\frac{1}{2} \sum_{i,j} (|Z_{i,j}| + Z_{i,j}) (d_i^2 + d_j^2)
~\geq~
\sum_{i,j} |Z_{i,j}| \, d_i d_j.
\end{equation}
Since $|Z_{i,j}| + Z_{i,j} \geq 0$, the AM-GM inequality
implies that the left-hand side is at least
\[
\sum_{i,j} (|Z_{i,j}| + Z_{i,j}) d_i d_j
~=~ \sum_{i,j} |Z_{i,j}| \, d_i d_j \:+\: d \transpose Z d.
\]
Since $Z \in \Psd$, this implies \eqref{eq:moregeneral2}.
\end{proofof}

\section{Connections to the Kadison-Singer Problen}
\AppendixName{kadison}

The Kadison-Singer problem, which dates back to 1959, is an important, and until very recently unsolved, question in operator theory.
The importance of this question has become increasingly apparent in recent
years as it is now known to be equivalent, or closely related, to
numerous conjectures in disparate areas of mathematics \cite{CT}.
In a very recent breakthrough, Marcus, Spielman and Srivastava~\cite{MSS} positively resolved the Kadison-Singer problem. 
\onote{Hmm, maybe this is only related to Anderson's paving conjecture?}
More precisely, they proved the following strong form of Weaver's conjecture~\protect\cite[Conjecture $\text{KS}_2$ and Theorem 2]{Weaver}:

\begin{theorem}[\cite{MSS}]\label{thm:MSS} 
Let $\epsilon > 0$, and $u_1,\ldots,u_m \in \bC^n$ such that
$\norm{u_i} \leq \epsilon$ for all $i$, and $\sum_i u_i u_i \transpose = I$.
Then there exists a partition of $[m]$ into $S_1, S_2$ such that for each $j \in \{1,2\}$, 
\begin{equation}\label{eq:MSS}
\sum_{i \in S_j} u_i u_i \transpose  ~\leq~ \tfrac12(1 + \sqrt{2\epsilon})^2.
\end{equation}
\end{theorem}

It is well-known that, given a strong discrepancy result such as \eqref{eq:MSS}, an iterative
argument yields a sparse object that gives a good approximation.
See, e.g., Rudelson \cite{RudelsonOrthogonal}. For the sake of completeness, we include here a detailed
argument that \Theorem{MSS} implies the existence of $O(1/\conduc)$-spectrally-thin trees.

First, the following corollary of \Theorem{MSS} will be convenient for induction purposes.

\begin{corollary}
\label{cor:strongKS}
There exists a constant $C \geq 1$ such that the following is true.
Let $v_1,\ldots,v_m \in \bR^n$ be such that
$\alpha I \preceq \sum_i v_i v_i \transpose \preceq \beta I$
and $\norm{v_i}^2 = \delta := n/m$ for all $i$.
Suppose that $\alpha \in [1/2,1]$ and $\beta \in [1,2]$.
Then there exists $S \subseteq [m]$ satisfying
\begin{equation}\label{eq:strongKS}
(\alpha - C \sqrt{\delta}) I
~\preceq~ 2 \sum_{i \in S} v_i v_i \transpose
~\preceq~ (\beta + C \sqrt{\delta}) I.
\end{equation}
\end{corollary}
\begin{proof}
Let $\alpha, \beta, \delta, v_1,\ldots,v_m$ be as in the statement of \Corollary{strongKS}.
Note that $\delta \leq 1$, since $m \geq n$. \onote{Say why?}
Letting $M = \sum_i v_i v_i \transpose$, we see that $\norm{M^{-1}} \leq \alpha^{-1}$.
Define $u_i = M^{-1/2} v_i$.
Then
$$
\sum_i u_i u_i \transpose
~=~ M^{-1/2} \Big(\sum_i v_i v_i \transpose\Big) M^{-1/2}
~=~ I
$$
and
$$
\norm{u_i}^2 ~\leq~ \norm{M^{-1}} \norm{v_i}^2 ~\leq~ \alpha^{-1} \delta ~=:~ \epsilon.
$$
Applying \Theorem{MSS} of Marcus et al.~\cite{MSS}, we deduce \eqref{eq:MSS}, and hence (since $\epsilon \leq 2$)
\[ 2\sum_{i \in S_j} u_i u_i \transpose ~\preceq~ 1 + 4\sqrt{2\epsilon} \qquad \text{for } j \in \{1,2\}. \]
Consequently,
$$
2 \sum_{i \in S_1} v_i v_i \transpose
~\preceq~ (1 + 4 \sqrt{2\epsilon}) M
~\preceq~ (1 + 4 \sqrt{2\epsilon}) \beta
~\preceq~ \beta + 16\sqrt{\delta}
$$
by the hypotheses $\alpha \in [1/2,1]$ and $\beta \in [1,2]$.

Observing that 
\[ 2\sum_{i \in S_1} u_i u_i\transpose = 2I - 2\sum_{i \in S_2} u_i u_i\transpose \succeq 1 - 4\sqrt{2\epsilon}, \]
we similarly obtain
$$
2 \sum_{i \in S_1} v_i v_i \transpose
~\succeq~ (1 - 4 \sqrt{2\epsilon}) M
~\succeq~ (1 - 4 \sqrt{2\epsilon}) \alpha
~\succeq~ \alpha - 4 \sqrt{2\delta}.
$$
Thus taking $S=S_1$, we see that \eqref{eq:strongKS} holds with $C=16$.
\end{proof}

We may now prove \Theorem{strongIsotropic} by an application of Corollary~\ref{cor:strongKS}.
By an argument similar to the proof of \Theorem{conductanceweightedtree},
this implies the existence of $O(1/\conduc)$-spectrally-thin trees.

\begin{claim}
\ClaimName{KSimplIsotropic}
\label{KSIMPLISOTROPIC}
Corollary~\ref{cor:strongKS} implies \Theorem{strongIsotropic}.
\end{claim}
\begin{proof}
\newcommand{\vv}[2]{v_{#1}^{(#2)}}
\newcommand{\al}{\alpha}
\newcommand{\be}{\beta}
\newcommand{\ga}{\gamma}

As in \Theorem{strongIsotropic}, 
let $w_1,\ldots,w_m \in \bR^n$ satisfy $\norm{w_i} = 1$ for all $i$.
Let $p_1,\ldots,p_m$ be a probability distribution on these vectors
such that the covariance matrix is $\sum_i p_i w_i w_i \transpose = I/n$.

Without loss of generality, we may assume $p_i = 1/m ~\forall i$.
To see this, suppose that $p_1,\ldots,p_m$ are rational numbers of the form $q_i/M$
where $q_1,\ldots,q_m,M$ are nonnegative integers.
Then we may replace each $w_i$ with $q_i$ copies of itself.
The uniform distribution on the resulting vectors still has covariance matrix $I/n$.
Proving \Theorem{strongIsotropic} for the resulting vectors establishes
the theorem for the original vectors under distribution $p$.
If $p_1,\ldots,p_m$ are irrationals, we may approximate them by rationals
while introducing vanishing error.

Define $v_i = \sqrt{n/m} \cdot w_i$, so that $\norm{v_i}^2 = n/m =: \delta_0$ for all $i$. 
We will iteratively construct sets $S_t \subseteq [m]$, with $S_0 = [m]$.
Let $C$ be as in \Corollary{strongKS}.
Define $\al_0 = \be_0 = m$, and then inductively
\begin{align*}
    \al_{t+1} &= \al_t - C (2^tn)^{1/2} (\be_t)^{1/2}\\
    \be_{t+1} &= \be_t + C (2^tn)^{1/2} (\be_t)^{1/2}.
\end{align*}

Let $\epsilon \in (0,1]$ be a small constant to be chosen in a moment, and let 
\[ T = \max \Bigl\{\, t \,:\,  C \sum_{j=0}^{t-1} (2^jn/m)^{1/2} \leq \epsilon/2\, \Bigr\}. \] 
\todo{Perhaps state the value of $T$ and forward pointer to \eqref{eq:T}.}

This choice of $T$ is motivated by the following:
\begin{claim}\label{clm:abbounds}
    For all $t \leq T$, $\be_t \leq m(1+\epsilon)$ and $\al_t \geq m(1-\epsilon)$.
\end{claim}
\begin{proof}
    For $0 \leq t < T$,
    \[ \be_{t+1} ~=~ \be_{t}(1 + C (2^{t}n / \be_{t})^{1/2}) ~\leq~ \be_{t}(1 + C (2^{t}n/m)^{1/2}). \]
So
\begin{align*} \be_t &~\leq~ m\prod_{j=0}^{t-1}\Bigl(1+ C (2^jn/m)^{1/2}\Bigr)\\
                        &~\leq~ m \exp\Bigl(C \sum_{j=0}^{t-1} (2^jn/m)^{1/2}\Bigr)\\
                        &~\leq~ m\exp(\epsilon/2)\\
                        &~\leq~ m(1+\epsilon).
\end{align*}
Note that 
\[ \al_0 - \al_t ~=~ \be_t - \be_0 \]
and so since $\be_t \leq m(1+\epsilon)$, $\al_t \geq m(1-\epsilon)$.
\end{proof}
Note that since $\sum_{j=0}^{T-1} (2^jn/m)^c = \Theta((2^Tn/m)^{1/2})$,  we have that 
\begin{equation}\label{eq:T}
2^T = \Theta\left(\frac{m}{n} \epsilon^{2}\right). 
\end{equation}
So we may choose $\epsilon \in (0,1/3]$ to be a constant sufficiently small so that 
\begin{equation}\EquationName{boundedByd}
    (1-\epsilon)^{-1}2^T (n/m) ~\leq~ d.
\end{equation}

Our first goal will be to show inductively that for all $t \leq T$, there exists a set $S_t \subseteq [m]$ so that
\begin{equation}\label{eq:indass}
\al_t ~\preceq~ m2^t \sum_{i \in S_t} v_iv_i\transpose ~\preceq~ \be_t.
\end{equation}
Note that this is true for $t=0$ by assumption.

It will be convenient to define $\ga_t = 2^t|S_t|$.
Suppose \eqref{eq:indass} holds for some particular $t < T$.
Define \[\vv{i}{t} ~=~ v_i \cdot \sqrt{m2^t/\ga_t}, \]
so that $\lVert\vv{i}{t}\rVert^2 = \frac{n}{|S_t|} =: \delta_t$ for all $t$. 
Then just by scaling,
\[ \al_t / \ga_t ~\preceq~ \sum_{i \in S_t} \vv{i}{t}(\vv{i}{t})\transpose ~\preceq~ \be_t / \ga_t. \]
Taking a trace yields $n \al_t/\ga_t \leq n \leq n\be_t/\ga_t$, i.e., 
\begin{equation}\EquationName{alphagammabeta}
\al_t ~\leq~ \ga_t ~\leq~ \be_t.
\end{equation}
By \eqref{eq:alphagammabeta}, \Claim{abbounds} and \eqref{eq:boundedByd}, we have
\begin{gather*}
1/2 ~\leq~ \frac{1-\epsilon}{1+\epsilon} ~\leq~ \frac{\al_t}{\be_t} ~\leq~ \frac{\al_t}{\ga_t} ~\leq~ 1 \\
1 ~\leq~ \frac{\be_t}{\ga_t} ~\leq~ \frac{\be_t}{\al_t} ~\leq~ \frac{1+\epsilon}{1-\epsilon} ~\leq~ 2 \\
\delta_t ~=~ 2^tn/\ga_t ~\leq~ 2^Tn/\al_t ~\leq~ 2^Tn/(m(1-\epsilon)) ~\leq~ d.
\end{gather*}
Now apply \Corollary{strongKS} with
$S_t$ instead of $[m]$,
$\vv{i}{t}$ instead of $v_i$,
$\al_t / \ga_t$ instead of $\alpha$, 
$\be_t / \ga_t$ instead of $\beta$, and
$\delta_t$ instead of $\delta$.
The hypotheses of \Corollary{strongKS} are satisfied, so it follows that there is a set $S_{t+1}
\subseteq S_t$ with
\[
   \al_t / \ga_t - C \delta_t^{1/2}
 ~\preceq~ 2\sum_{i \in S_{t+1}} \vv{i}{t}(\vv{i}{t})\transpose
 ~\preceq~ \be_t / \ga_t + C \delta_t^{1/2}.
\]
Rewriting in terms of the original $v_i$'s, we obtain
\[
    \al_t - C \ga_t \delta_t^{1/2}
    ~\preceq~ 2^{t+1}m \sum_{i \in S_{t+1}} v_iv_i\transpose
    ~\preceq~ \be_{t} + C \ga_t \delta_t^{1/2}.
\]
Now 
\[ \ga_t\delta_t^{1/2} ~=~ \ga_t (n/|S_t|)^{1/2} ~=~ (2^tn)^{1/2}(\ga_t)^{1/2} ~\leq~ (2^tn)^{1/2} (\be_t)^{1/2}. \]
Hence
\[ \al_{t+1} ~\preceq~ 2^{t+1}m \sum_{i \in S_{t+1}} v_iv_i\transpose ~\preceq~ \be_{t+1}, \]
and the inductive step is achieved.

\medskip

From \Claim{abbounds} and \eqref{eq:indass} for $t=T$, we deduce that
\[ m(1-\epsilon) ~\preceq~ 2^Tn \sum_{i \in S_T} w_iw_i\transpose ~\preceq~ m(1+\epsilon); \]
hence by \eqref{eq:T} and since $\epsilon$ is a constant,
\[         \Theta\big( (1-\epsilon) \epsilon^{-2} \big)
 ~\preceq~ \sum_{i \in S_T} w_iw_i\transpose
 ~\preceq~ \Theta\big( (1+\epsilon) \epsilon^{-2} \big)
 ~=~ \Theta(1). \]

The left inequality implies that $\setst{ w_i }{ i \in S_T }$ spans $\bR^n$.
To conclude,
select an arbitrary basis $B \subseteq S_T$; then
\[ \sum_{i \in B} w_iw_i\transpose ~\preceq~ \sum_{i \in S_T} w_iw_i\transpose ~\preceq~ \Theta(1), \]
and so \Theorem{strongIsotropic} holds.

\end{proof}

\end{document}